\documentclass[11pt]{article}

\usepackage{xcolor, verbatim}

\usepackage[pagebackref]{hyperref}

\renewcommand{\backref}[1]{}

\renewcommand{\backrefalt}[4]{%
\ifcase #1
\or $^{#2}$%
\else $^{#2}$%
\fi}
\usepackage{kpfonts}
\hypersetup{
    colorlinks,
    linkcolor={black},
    citecolor={black},
}

\usepackage{fullpage}
\usepackage{amsmath,amsthm,amsfonts,dsfont}
\usepackage{amssymb,latexsym,graphicx}
\usepackage{palatino}
\usepackage{mathpazo}
\usepackage{stmaryrd}
\usepackage{mathtools}
\usepackage{hyperref}
\usepackage{xspace}
\usepackage{float}

\newtheorem{theorem}{Theorem}[section]
\newtheorem*{remark}{Remark}
\newtheorem{proposition}[theorem]{Proposition}
\newtheorem{lemma}[theorem]{Lemma}

\newtheorem{claim}[theorem]{Claim}

\newtheorem{definition}[theorem]{Definition}
\newtheorem{question}[theorem]{Question}

  \makeatletter
  \newtheorem*{rep@theorem}{\rep@title}
  \newcommand{\newreptheorem}[2]{%
  \newenvironment{rep#1}[1]{%
  \def\rep@title{#2 \ref{##1}}%
  \begin{rep@theorem}}%
  {\end{rep@theorem}}}
  \makeatother
  
  \newreptheorem{lemma}{Lemma}
  \newreptheorem{theorem}{Theorem}
  \newreptheorem{corollary}{Corollary}
  \newreptheorem{proposition}{Proposition}
  \newreptheorem{conjecture}{Conjecture}  

\usepackage{algorithm}
\usepackage{algorithmicx}
\usepackage{algpseudocode}
\usepackage{mathtools}

\DeclareMathOperator{\db}{db}

\DeclareMathOperator*{\E}{\mathbb{E}}
\newcommand{\C}{\ensuremath{\mathbb{C}}}
\newcommand{\N}{\ensuremath{\mathbb{N}}}
\newcommand{\F}{\ensuremath{\mathbb{F}}}
\newcommand{\R}{\ensuremath{\mathbb{R}}}
\newcommand{\Z}{\ensuremath{\mathbb{Z}}}

\newcommand{\ket}[1]{|#1\rangle}
\newcommand{\bra}[1]{\langle#1|}

\newcommand{\braket}[2]{\langle#1\vert#2\rangle}

\newcommand{\eps}{\varepsilon} 
\renewcommand{\epsilon}{\varepsilon}

\DeclareMathOperator{\sign}{sgn}

\renewcommand{\epsilon}{\varepsilon} 
\newcommand{\T}{\mathbb T}
\newcommand{\norm}[1]{\lVert#1\rVert}

\newcommand{\bset}[1]{\{0,1\}^{#1}} 

\renewcommand{\Pr}{\mbox{\rm Pr}}

\newcommand{\Exp}{\mathbf{E}}

\newcommand{\beq}{\begin{equation}}
\newcommand{\eeq}{\end{equation}}
\newcommand{\beqn}{\begin{equation*}}
\newcommand{\eeqn}{\end{equation*}}
\newcommand{\beqr}{\begin{eqnarray}}
\newcommand{\eeqr}{\end{eqnarray}}
\newcommand{\beqrn}{\begin{eqnarray*}}
\newcommand{\eeqrn}{\end{eqnarray*}}

  \floatstyle{ruled}
  \newfloat{algorithm}{h}{lob}[section]
  \floatname{algorithm}{Algorithm}

\newif\ifnotes\notesfalse

\ifnotes
\usepackage{color}
\definecolor{mygrey}{gray}{0.50}
\newcommand{\notename}[2]{{\textcolor{mygrey}{\footnotesize{\bf (#1:} {#2}{\bf ) }}}}
\newcommand{\noteswarning}{{\begin{center} {\Large WARNING: NOTES ON}\end{center}}}

\else

\newcommand{\notename}[2]{{}}
\newcommand{\noteswarning}{{}}

\fi

\title{Bounding quantum-classical separations for classes of nonlocal games}
\author{
    \begin{tabular}{c}
        Tom Bannink\footnote{CWI, QuSoft, Science Park 123, 1098 XG Amsterdam, Netherlands. Supported by the Gravitation-grant NETWORKS-024.002.003 from the Netherlands Organisation for Scientific Research~(NWO).}\\
        \footnotesize{\texttt{bannink@cwi.nl}}
    \end{tabular}
    \and
    \begin{tabular}{c}
        Jop Bri\"et\textsuperscript{*}\footnote{Additionally supported by a VENI grant.}\\
        \footnotesize{\texttt{j.briet@cwi.nl}}
    \end{tabular}
    \and
    \begin{tabular}{c}
        Harry Buhrman\textsuperscript{*}\footnote{University of Amsterdam. Additionally supported by NWO Gravitation grant QSC, also supported by EU grant QuantAlgo.}\\
        \footnotesize{\texttt{buhrman@cwi.nl}}
    \end{tabular}
    \and
    \begin{tabular}{c}
        Farrokh Labib\textsuperscript{*}\\
        \footnotesize{\texttt{labib@cwi.nl}}
    \end{tabular}
    \and
    \begin{tabular}{c}
        Troy Lee\footnote{Centre for Quantum Software and Information,
School of Software, Faculty of Engineering and Information Technology,
University of Technology Sydney, Australia.
Part of this work was done while at the School for Physical and Mathematical Sciences, Nanyang Technological University and the Centre 
for Quantum Technologies, Singapore, supported by the Singapore National 
Research Foundation under NRF RF Award No. NRF-NRFF2013-13.}\\
        \footnotesize{\texttt{troyjlee@gmail.com}}
    \end{tabular}
}
\date{}

\begin{document}
\maketitle

\noteswarning

\begin{abstract}
    We bound separations between the entangled and classical values for several classes of nonlocal $t$-player games.
Our motivating question is whether there is a family of $t$-player XOR games for which the entangled bias is $1$ but for which the classical bias goes down to $0$, for fixed~$t$.
Answering this question would have important consequences in the study of multi-party communication complexity, as a positive answer would imply an unbounded separation between randomized communication complexity with and without entanglement.
Our contribution to answering the question is identifying several general classes of games for which the classical bias can not go to zero when the entangled bias stays above a constant threshold.
This rules out the possibility of using these games to answer our motivating question.
A previously studied set of XOR games, known not to give a positive answer to the question, are those for which there is a quantum strategy that attains value 1 using a so-called Schmidt state.
We generalize this class to mod-$m$ games and show that their classical value is always at least $\frac{1}{m} + \frac{m-1}{m} t^{1-t}$.
Secondly, for free XOR games, in which the input distribution is of product form, we show $\beta(G) \geq \beta^*(G)^{2^t}$ where~$\beta(G)$ and~$\beta^*(G)$ are the classical and entangled biases of the game respectively.
We also introduce so-called line games, an example of which is a slight modification of the Magic Square game, and show that they can not give a positive answer to the question either.
Finally we look at two-player unique games and show that if the entangled value is $1-\epsilon$ then the classical value is at least $1-\mathcal{O}(\sqrt{\epsilon \log k})$ where $k$ is the number of outputs in the game.
Our proofs use semidefinite-programming techniques, the Gowers inverse theorem and hypergraph norms.

\end{abstract}

\section{Introduction} \label{sec:intro}
The study of multiplayer games has been extremely fruitful in theoretical computer science across diverse areas including 
the study of complexity classes \cite{BGKW:1988}, hardness of approximation \cite{Khot:2002}, and communication 
complexity \cite{KLLRX:2015}.  They are also a great framework in which to study 
Bell inequalities \cite{Bell:1964} and analyze the nonlocal properties of 
entanglement.  A particularly simple 
kind of multiplayer game is an XOR game.  An XOR game 
$G = (f,\pi)$ between $t$-players is defined by a function 
$f: X_1 \times X_2 \times \cdots \times X_t \rightarrow \{0,1\}$ 
and a probability distribution $\pi$ over $X_1 \times \cdots \times X_t$.  An 
input 
$(x_1, \ldots, x_t) \in X_1 \times \cdots \times X_t$ is chosen by a referee
according to $\pi$, who then gives $x_i$ to player $i$.  Without communicating, player $i$ then 
outputs a bit $a_i \in \{0,1\}$ with the collective goal of the players being 
that $a_1 \oplus \cdots \oplus a_t = f(x_1, \ldots, x_t)$. In a classical XOR game, the players' 
strategies are deterministic.  In an XOR game with entanglement, players are allowed to share a quantum 
state and make measurements on this state to inform their outputs.  

As players can always win an XOR game with probability at least~$\frac{1}{2}$, it is common to 
study the \emph{bias} of an XOR game, the probability of winning minus the probability of losing.  We use 
$\beta(G)$ to denote the largest bias achievable by a classical protocol for the game $G$, and 
$\beta^*(G)$ to denote the best bias achievable by a protocol using 
shared entanglement for the game $G$.  

Our motivating question in this paper is: 
\begin{question}
\label{quest:motivation}
\emph{Is there a family of $t$-player XOR games $(G_n)_{n\in \N}$ such that $\beta^*(G_n) =1$ and 
$\beta(G_n) \rightarrow 0$ as $n \rightarrow \infty$?} 
\end{question}

This question has important implications for multi-party communication complexity.  For a function 
$f: X_1 \times \cdots X_t \rightarrow \{0,1\}$, let $R(f)$ denote the $t$-party randomized 
communication 
complexity of $f$ (in the number-in-the-hand model), and let $R^*(f)$ denote the 
$t$-party randomized communication complexity of $f$ 
where the parties are allowed to share entanglement.  A positive answer to 
Question~\ref{quest:motivation} gives a family of functions $(f_n)_{n\in \N}$ 
with $R^*(f_n) = O(1)$ and $R(f_n) = \omega(1)$, i.e.\ an unbounded separation 
between these two communication models.  

In the reverse direction, a family of functions $(f_n)_{n\in \N}$ 
with $R^*(f_n) = O(1)$ and $R(f_n) = \omega(1)$ gives a family of games 
$G_n = (f_n, \pi_n)$ with $\beta^*(G_n) \ge c$ for some constant $c$ and 
$\beta(G_n) \rightarrow 0$ as $n \rightarrow \infty$.  Thus there is a very 
close connection between Question~\ref{quest:motivation} and the existence of an unbounded 
separation between randomized communication complexity with and without 
entanglement.

For the two-player case, it is known that the answer to Question~\ref{quest:motivation} is 
negative.  It was observed by Tsirelson~\cite{Tsirelson:85b} that Grothendieck's 
inequality~\cite{Grothendieck:1953}, a fundamental result from Banach space theory, is equivalent to 
the assertion that $\beta^*(G) \le K_G \cdot \beta(G)$, where 
$K_G \le 1.78\ldots$ \cite{Krivine:1977, BMMN:2011} is Grothendieck's 
constant.

Linial and Shraibman \cite{LinalShraibman:2009} and Shi and Zhu \cite{ShiZhu:2008} realized that 
the XOR bias of a game $(f, \pi)$ can be used to lower bound the communication 
complexity of $f$, 
both in the randomized setting and the setting with entanglement.  Together with 
Grothendieck's 
inequality they used this to show that $R(f) =O(2^{2 R^*(f)})$ for any partial two-party 
function $f$.  Thus in the two-party case an unbounded communication separation 
is not possible between the randomized model with and without entanglement.  Raz 
has given 
an example of a partial function $f$ with $R(f) = 2^{\Omega(R^*(f))}$ 
\cite{Raz:1999}, thus the upper bound of Linial-Shraibman and Shi-Zhu is 
essentially optimal.

In the case of three or more parties, Question~\ref{quest:motivation} and the corresponding question of an  unbounded separation between the entangled and non-entangled communication complexity models remain open.
A striking result of Per\'ez-Garc\'ia et al.~\cite{Perez-Garcia:2008} shows that 
there is no analogue of Grothendieck's inequality in the three-player setting.
In particular, they showed that there exists an infinite family of three-player XOR games 
$(G_n)_{n\in \N}$ with the property that the ratio of the entangled and classical biases of~$G_n$ goes 
to infinity with~$n$.
This result was later quantitatively improved by Bri\"{e}t and 
Vidick~\cite{BrietVidick:2013}.
Both results rely crucially on non-constructive (probabilistic) methods, and in both 
separating examples the entangled bias $\beta^*(G_n)$ also goes to zero with increasing~$n$.  
These works leave open the question, posed explicitly 
in~\cite{BrietVidick:2013}, of whether there is such a family of games in which 
the entangled bias does not vanish with~$n$, but instead stays above a fixed 
positive threshold while the classical bias decays to zero.  Crucially, having a 
separation in XOR bias where $\beta^*(G_n)$ remains constant is what is needed 
to also obtain an unbounded separation between randomized communication 
complexity with and without entanglement.

\paragraph*{Our contribution to answering Question~\ref{quest:motivation}}
One approach to Question~\ref{quest:motivation} is to look at different classes of games 
and identify which ones could possibly lead to a positive answer.  

Per\'ez-Garc\'ia et al.~\cite{Perez-Garcia:2008} show that in any XOR game 
where the entangled strategy uses a GHZ state, there is a bounded gap 
between the classical and entangled bias: namely, the bias with a GHZ state in 
a $t$-player XOR game $G$ is at most $K_G (2\sqrt{2})^{t-1}\beta(G)$.  This 
bound is essentially tight as there are examples of $t$-player XOR games 
achieving a ratio between the GHZ state bias and classical bias of 
$\frac{\pi}{2}^t$ \cite{Zuk:1993}.  Bri\"{e}t et al.~\cite{Briet:2013} later 
extended 
the Grothendieck-
type inequality of Per\'ez-Garc\'ia et al.\ to a larger class of entangled 
states called Schmidt states (see Equation~\ref{eq:schmidt}).  Thus any 
game where there is a perfect strategy where the players share a Schmidt state 
cannot give a positive answer to Question~\ref{quest:motivation}.

Watts et al.\ \cite{Watts:2018} recently investigated 
Question~\ref{quest:motivation} and found that a $t$-player XOR game $G$ that 
is symmetric, i.e.\ invariant under the renaming of players, and where 
$\beta^*(G)=1$ always has a perfect entangled strategy where the players 
share a GHZ state.  Thus symmetric games also cannot give a positive 
answer to Question~\ref{quest:motivation}.

We further study games that have a perfect strategy where players share a 
GHZ or Schmidt state.  We do this for a generalization of XOR games called 
mod $m$ games.  In a mod $m$ game the players output an integer between $0$ and 
$m-1$ and the goal is for the sum of the outputs mod $m$ to equal a target 
value determined by their inputs.  We show that the classical advantage over 
random guessing is at least $\frac{m-1}{m}t^{1-t}$ in any $t$ player mod 
$m$ game that can be won perfectly by sharing a Schmidt state (see 
Theorem~\ref{thm:perfectghz}).

We show this by introducing \emph{angle games}, a class of games that can be 
won perfectly sharing a GHZ state and are the \emph{hardest} of all such 
games.  Thus a classical strategy in an angle game can be used to lower bound 
the winning probability of any mod $m$ game that has a perfect 
Schmidt state strategy.  

For small values of $t$ we can directly analyze angle games to give 
bounds that are sometimes tight.
One interesting consequence of our result is the following.  
The Mermin game $G$ is a 
three-party XOR game where by sharing a GHZ state players can play perfectly, 
$\beta^*(G)=1$, while classically $\beta(G)= \frac{1}{2}$.  We show that this 
is the maximal possible separation of any 3-party XOR game where $\beta^*(G)=1$ 
via a GHZ state.  In particular, this means that when one looks at the XOR 
repetition of the Mermin game the 
classical bias \emph{does not go down at all}.  

We rule out other types of games that could positively answer 
Question~\ref{quest:motivation} as well.  A $t$-player \emph{free} XOR game 
$G = (f,\pi)$ is a game where $\pi$ is a product distribution.  For such 
games we show that $\beta(G) \ge \beta^*(G)^{2^t}$, and thus they cannot be 
used for a positive answer to Question~\ref{quest:motivation}.

Another class of XOR games we consider are \emph{line games}, where the 
questions asked to the players are related by a geometric property.  An 
example of a line game is a slight modification of the Magic Square game \cite{Ito:2008}.  We show that 
line games cannot give a positive answer to Question~\ref{quest:motivation} 
either.

Finally, we look at extensions of Question~\ref{quest:motivation} beyond XOR 
games to more general classes of games like unique games \cite{Khot:2002}, which 
have been deeply studied because of their application in hardness of 
approximation.  For unique games we show that in fact if 
there is strategy with entanglement that can win a unique game perfectly, then 
there is a perfect classical strategy as well.  This can be compared with the 
result of Cleve et al.\ \cite{CHTW:2004} that if a two-player game with binary 
outputs has a perfect strategy with entanglement then it also has a perfect 
classical strategy.
More generally, we show that 
if the winning probability with entanglement is $1-\epsilon$ in a unique game 
with $k$ outputs, then there is 
a classical strategy that wins with probability $1- C\sqrt{\epsilon \log k}$.

In the next subsections, we discuss our results in more detail.
  
\subsection{Perfect Schmidt strategies for MOD games} \label{sec:introperfectghz}
A MOD-$m$ game is a generalization of XOR games to non-binary outputs.
A nonlocal game is a MOD-$m$ game if the players are required to answer with integers from 0 to $m-1$, and win if and only if the sum of their answers modulo~$m$ equals the target value determined 
by their inputs. We denote the optimal winning probability using classical strategies by $\omega(G)$, and we write $\omega^*(G)$ for the entangled winning probability. Random play in such a game ensures that the players can always win with probability at least~$\frac{1}{m}$.  As with XOR games, in a MOD-$m$ game one 
often considers the \emph{bias} given by the maximum amount by which the value can exceed~$\frac{1}{m}$, scaled to be in the $[0,1]$ range. The bias is $\beta(G) = \frac{m}{m-1}(\omega(G)-\frac{1}{m})$, and similar for the entangled version. This generalizes the definition given for XOR 
games above.

Define a $t$-partite \emph{Schmidt state} as a $t$-partite quantum state that can be written in the form
\begin{align}
\label{eq:schmidt}
    \ket{\psi} = \sum_{i=0}^{d-1} c_i \ket{e^{(1)}_i}\ket{e^{(2)}_i}\cdots\ket{e^{(t)}_i} ,
\end{align}
where $c_i > 0$ and where the $\ket{e^{(j)}_i}$ ($i = 0,1,...,d-1$) are orthogonal vectors in the $j$-th system. For $t=2$ any state can be written this way, something commonly known as the Schmidt decomposition. Note that the well-known GHZ state is a Schmidt state where all the $c_i$ are equal to $1/\sqrt{d}$.
In the context of nonlocal games, define a \emph{Schmidt strategy} as a quantum strategy that uses (only) a Schmidt state. We say a strategy is \emph{perfect} if it achieves winning probability 1.

We consider $t$-player MOD-$m$ games for which there is a perfect Schmidt strategy (``perfect Schmidt games'') and for such games we give lower bounds on the classical winning probabilities.
One particular set of games with this property is described by Boyer~\cite{Boyer:2004}. Their entangled value is 1 but their classical value goes to 0 as the number of players goes to infinity.
The authors of~\cite{Watts:2018} define a closely-related class of games called \emph{noPREF games}. This set of games is equal to the set of perfect Schmidt games when $m=2$ and the distribution on the inputs is uniform. In \cite{Watts:2018} it is shown that checking whether a game is in this class can be done in polynomial time. Furthermore, for \emph{symmetric} $t$-player XOR games they show that a game has entangled value 1 if and only if it falls in this class of perfect Schmidt games. They also provide an explicit non-symmetric XOR game with entangled value 1 that is not in this class.
We introduce a $t$-player MOD-$m$ game called the \emph{uniform angle game}, denoted $\mathrm{UAG}_{t,m}$ (defined in Section~\ref{sec:uag}, Definition~\ref{def:uag}) for which there is a perfect Schmidt strategy and show a lower bound on the classical winning probability.
\begin{theorem} \label{thm:perfectghz}
    Any $t$-player MOD-$m$ game $G$ with perfect Schmidt strategy satisfies $\omega(G) \geq \omega(\mathrm{UAG}_{t,m})$. Furthermore we have $\beta(\mathrm{UAG}_{t,m}) \geq t^{1-t}$.
\end{theorem}
For $t=3,m=2$ ($3$-player XOR games) we have $\omega(\mathrm{UAG}_{3,2}) = 3/4$. In Section~\ref{sec:perfectghz} we provide bounds on $\omega(\mathrm{UAG}_{t,m})$ for other values of $t,m$.

Let the inputs to a game come from a set $X=X_1 \times X_2 \times ... \times X_t$ where $X_i$ is the set of inputs for the $i$-th player. We say a game is \emph{total} when all elements of $X$ have a non-zero probability of being asked (sometimes also called having \emph{full support}), similar to total functions in the setting of communication complexity. On the other hand, we say that a game has a \emph{promise} on the inputs when it is not total. For the class of perfect Schmidt games we show that total games are trivial.

\begin{lemma} \label{lemma:perfectghztotal}
    When a $t$-player MOD-$m$ game $G$ with perfect Schmidt strategy is total then $\omega(G)=1$.
\end{lemma}

\subsection{Free XOR games}

In this subsection we identify two types of games, namely \emph{free games} and \emph{line games}, for which either the ratio of the entangled and classical biases is small, or the entangled bias itself is small.  Thus these games 
will not be able to give a positive answer to Question~\ref{quest:motivation}.
Free games are a general and natural class of games in which the players' questions are independently distributed.
Line games appear to be less studied (see below for their definition), but turn out to be relevant in the context of parallel repetition (also see below).
The main idea behind these results is that a large entangled bias implies that the games are in a sense far from random.
This is quantified by the magnitude of certain norms of the game tensors.
The particular norms of interest here are related to norms used in Gowers' celebrated hypergraph- and Fourier-analytic proofs of Szemer\'edi's Theorem.
A crucial fact of these norms is that they are large if and only if there is ``correlation with structure'', the opposite of what one would expect from randomness.
We show that this structure can be turned into good classical strategies, thus establishing a relationship between the entangled and classical biases.


\begin{theorem}[Polynomial bias relation for free XOR games]\label{thm:xor_biases}
For integer~$t\geq 2$ and any free $t$-player XOR game with entangled bias~$\beta$, the classical bias is at least $\beta^{2^t}$.
\end{theorem}

This result may be considered as an analogue of a well-known result on quantum query algorithms for total functions.
It is shown in~\cite{Beals:2001} that the bounded-error quantum and classical query complexities of total functions are polynomially related.

\subsection{Line games}

\emph{Line games} are not free, but have a  simple geometric structure.
For a finite field~$\F$ of characteristic at least~$t$ and positive integer~$n$, a $t$-player line game is given by a map~$\tau:\F^n\to\bset{}$.
In the game, the referee independently samples two uniformly random points~$x,y\in \F^n$ and sends the point $x + (i-1)y$ to the $i$th player.
The players win the game if and only if the XOR of their answers equals~$\tau(y)$.
In other words, the players' questions correspond to consecutive points (or an arithmetic progression) on a random affine line through~$\F^n$ and the winning criterion depends only on the direction of the line.
Refer to this as a line games \emph{over~$\F^n$}.

A small example of a line game can be obtained from a slight modification of the three-player Magic Square game, which was analyzed in~\cite{Ito:2008}.
The line game is played over the plane~$\F_3^2$ and the predicate is zero only on the horizontal lines (with $y \in \{(1,0), (2,0)\}$.
In the Magic Square game, the referee restricts only to horizontal and vertical lines.\footnote{Though this is not the typical description of the game, it is easily seen to be equivalent.}

\begin{theorem}\label{thm:line_games}
For any $\eps \in (0,1]$, integer~$t\geq 2$ and finite field~$\F$ of characteristic at least~$t$, there exists a~$\delta(\eps,t,\F) \in (0,1]$ such that the following holds. 
For any positive integer~$n$ and any $t$-player line game over~$\F^n$ with entangled bias~$\eps$, the classical bias is at least~$\delta(\eps,t,\F)$.
\end{theorem}

Note that in the above result, the value of the classical bias is independent of the dimension~$n$ of the vector space determining the players' question sets.

While it is not relevant to Question~\ref{quest:motivation}, the proof techniques used for Theorem~\ref{thm:line_games} allow us to prove a parallel repetition theorem for a class of games that include line games. It is known that the value of free games and so-called anchored games decays exponentially under parallel repetition. Dinur et al.~\cite{Dinur:2016} identified a general criterion of multi-player games to behave like this, encompassing free and anchored games.
They showed that it is sufficient for a certain graph that can be obtained from a game to be expanding, a well-known pseudorandom property that gives a measure of graph connectivity.
Line games do not belong to this class, as their graphs are not even connected.
However, we show that if a map~$\tau:\F^n\to\bset{n}$ is pseudorandom in a different sense, then a line game defined by~$\tau$ has exponential decaying value under parallel repetition.
More generally, we show that this is the case for a family of XOR games over an arbitrary finite abelian group~$\Gamma$. 
These games are given by a positive integer~$m$, a family of affine linear maps~$\psi_0,\dots,\psi_t:\Gamma^m\to \Gamma$ such that no two are multiples of each other, and a ``game map'' $\rho:\Gamma\to\bset{}$.
In the game, the referee samples a uniform random element~$x$ from~$\Gamma^m$ and sends the group element $\psi_i(x)$ to the $i$th player.
The winning criterion is given by~$\rho(\psi_0(x))$.
The relevant notion of pseudoranomness is quantified by the Gowers $t$-uniformity norm of the map~$(-1)^\rho: x\mapsto (-1)^{\rho(x)}$, denoted~$\|(-1)^{\rho}\|_{U^t}$.

\begin{lemma}\label{lemma:parallelRepetition}
Let~$m,t$ be positive integers and let~$\Gamma$ be a finite abelian group.
Let $\psi_0,\dots,\psi_t:\Gamma^m\to \Gamma$ be affine linear maps such that no two are multiples of each other and let $\rho:\Gamma\to\bset{}$.
Let~$G$ be the~$t$-player XOR game given by the system $\{\psi_0,\dots,\psi_t,\rho\}$.
Then, for every positive integer~$k$, 
\begin{align*}
    \omega(G^k) \leq \Big(\frac{1 + \|(-1)^\rho\|_{U^t}}{2}\Big)^k.
\end{align*}
\end{lemma}

\subsection{Unique games}

We know that the answer to Question~\ref{quest:motivation} is negative in the two-player case, but we can generalize the question by dropping the XOR restriction.
The set of XOR games is part of a larger class of games called unique games for which we investigate the relation between classical and entangled values. A two-player nonlocal game is a unique game if for every pair of questions, for every possible answer of the first player there is exactly one answer of the second player that lets them win, and vice versa. Stated differently, for every question there is a matching between the answers of the two players such that only the matching pairs of answers let the player win.

The Unique Games Conjecture (UGC) of Khot~\cite{Khot:2002} states that for any $\epsilon,\delta>0$, for any $k>k(\epsilon,\delta)$, it is NP-hard to distinguish instances of unique games with winning probability at least $1-\epsilon$ from those with winning probability at most $\delta$, where $k$ is the number of possible answers. This conjecture has important consequences because it implies several hardness of approximation results. For example, for the Max-Cut problem, Khot et al.~\cite{Khot:2007} showed that the UGC implies that obtaining an approximation ratio better than $\approx 0.878$ is NP-hard. Other results include inapproximability for Vertex Cover~\cite{Khot:2008} and graph coloring problems~\cite{Dinur:2009}.

Our results relate the quantum and classical winning probabilities in the regime of near-perfect play and are based on a result in~\cite{Charikar:2006}.

\begin{theorem} \label{thm:uniquegames}
    Let $\epsilon \geq 0$. There is an efficient algorithm that, given any two-player unique game with entangled value $1-\epsilon$, outputs a classical strategy with winning probability at least $1- C \sqrt{\epsilon \log k}$, where $C$ is a constant independent of the game.
\end{theorem}
Note that for $\epsilon=0$ this means a perfect quantum strategy implies a perfect classical strategy.
Furthermore, the above result only beats a trivial strategy when $\epsilon = \mathcal{O}(1/\log k)$.

Work in a similar direction includes~\cite{Kempe:2008}. They show that entangled version of the UGC is false, by providing an efficient algorithm that gives an explicit quantum strategy with winning probability at least $1-6\epsilon$ when the true entangled value is $1-\epsilon$. In the classical case, \cite{Charikar:2006} gives an algorithm that outputs a classical strategy with winning probability $1-\mathcal{O}(\sqrt{\epsilon \log k})$ when the true \emph{classical value} is $1-\epsilon$. We extend this result by showing that this classical strategy also does the job when, not the classical, but the entangled value is $1-\epsilon$.

\section{Techniques} \label{sec:techniques}

This section provides an overview of the proof techniques that we employed. We give sketches of the main ideas which are worked out in full detail in later sections.

\subsection{Reduction to angle games.}
To prove Theorem~\ref{thm:perfectghz} we introduce a new set of $t$-player MOD-$m$ games that we call \emph{angle games}. We define a particular angle game called the \emph{uniform angle game}, denoted by  $\mathrm{UAG}_{t,m}$ and show that it is the hardest of these games. In an angle game, players receive complex phases $e^{i\phi}$ (angles) satisfying a promise, and the winning answer depends only on the product of the inputs $e^{i\phi_1} \cdot e^{i\phi_2} \cdots e^{i\phi_t}$.
We prove the theorem by extracting from any perfect Schmidt strategy a set of complex phases that satisfy such a promise, and thereby reducing any such game to the $\mathrm{UAG}_{t,m}$ game. Let us sketch how this is accomplished. Assume that a perfect Schmidt strategy exists, and let $\{ P^{(j,x_j)}_1, ..., P^{(j,x_j)}_m \}$ be the projective measurement done by player $j$ on input $x_j$ so that $P^{(j,x_j)}_i$ corresponds to output $i$. Now define unitaries $U^{(j,x_j)} = \sum_i \omega_m^i P^{(j,x_j)}_i$, where $\omega_m=e^{2\pi i/m}$ is an $m$-th root of unity. Since the strategy is perfect we have for every input $(x_1,...,x_t)$ that
\begin{align*}
    \omega_m^{M(x_1,...,x_t)} &=  \bra{\psi} U^{(1,x_1)} \otimes U^{(2,x_2)} \otimes ... \otimes U^{(t,x_t)} \ket{\psi} .
\end{align*}
Using the definition of a Schmidt state, we show that this equality implies that these unitaries must be of a simple form and their entries satisfy the promise of an angle game.
We prove Theorem~\ref{thm:perfectghz} and Lemma~\ref{lemma:perfectghztotal} in Section~\ref{sec:perfectghz},
where we also provide classical strategies for the uniform angle game and show that these are tight in the case of $3$-player XOR games.

\subsection{Norming hypergraphs and quasirandomness.}
Our main tool for proving Theorem~\ref{thm:xor_biases} is a relation between the entangled and classical biases and a norm on the set of game tensors.
For $t$-tensors, this norm is given in terms of a certain $t$-partite $t$-uniform hypergraph~$H$. 
Recall that such a hypergraph consists of~$t$ finite and pairwise disjoint vertex sets $V_1,\dots,V_t$ and a collection of $t$-tuples $E(H) \subseteq V_1\times\cdots\times V_t$, referred to as the edge set of~$H$.
For a $t$-tensor $T \in \R^{n_1\times\cdots\times n_t}$, the norm has the following form:
\begin{equation}\label{eq:hypnorm}
\|T\|_H
=
\Big(
\Exp_{\phi_i:V_i\to [n_i]}\Big[ \prod_{(v_1,\dots,v_t)\in E(H)} T\big(\phi_1(v_1),\dots,\phi_t(v_t)\big)\Big]
\Big)^{\frac{1}{|E(H)|}},
\end{equation}
where the expectation taken with respect to the uniform distribution over all $t$-tuples of mappings $\phi_i$ from~$V_i$ to~$[n_i]$.
Expressions such as~\eqref{eq:hypnorm} play an important role in the context of graph homomorphisms~\cite{Borgs:2006}.
If~$T$ is the adjacency matrix of a bipartite graph with left and right node sets~$[n_1]$ and~$[n_2]$ respectively, then each product in~\eqref{eq:hypnorm} is~1 if and only if the maps $\phi_1$ and~$\phi_2$ preserve edges.

Criteria for~$H$ under which~\eqref{eq:hypnorm} defines a norm or a semi-norm were determined by Hatami~\cite{Hatami:2010, HatamiPhD} and Conlon and Lee~\cite{ConlonLee:2017}.
Famous examples of graph norms include the Schatten-$p$ norms for even~$p\geq 4$ (in which case~$H$ is a $p$-cycle) and a well-known family of hypergraph norms are the Gowers octahedral norms.
The latter were introduced for the purpose of quantifying a notion of quasirandomness of hypergraphs as an important part of Gowers' graph-theoretic proof of Szemer\'edi's theorem on arithmetic progressions.
Having large Gowers norm turns out to imply \emph{correlation with structure}, as opposed to quasirandomness.
This is true also for the norm relevant for our setting.
In particular, it turns out that the structure with which a game tensor correlates can be turned into a classical strategy for the game.
As such, a large norm of the game tensor implies a large classical bias of the game itself.
At the same time, we show that the entangled bias is bounded from above by the norm of the game tensor, provided the game is free.
Putting these observations together gives the proof of Theorem~\ref{thm:xor_biases}, which we give in Section~\ref{sec:hypergraphs}.

The particular hypergraph norm relevant in our setting was introduced in~\cite{Conlon:2012} and can be obtained recursively as follows.
Starting with a $t$-partite $t$-uniform hypergraph~$H$ with vertex set $V_1\cup\cdots\cup V_t$, write $\db_i(H)$ for the $t$-partite $t$-uniform hypergraph obtained by making two vertex-disjoint copies of~$H$ and gluing them together so that the vertices in the two copies of~$V_i$ are identified.
We obtain our hypergraph by starting with a single edge $e = (v_1,\dots,v_t)$ (and vertex sets of size~1), and applying this operation to all parts, forming the hypergraph $\db_1(\db_2(\dots \db_t(e)))$ with vertex sets of size~$2^{t-1}$ and~$2^t$ edges.
The fact that this hypergraph defines a norm via~\eqref{eq:hypnorm} was proved in~\cite{ConlonLee:2017}.

\subsection{Line games and Gowers uniformity norms.}
The proof of Theorem~\ref{thm:line_games} is based on two fundamental results from additive combinatorics: the generalized von Neuman inequality and the Gowers Inverse Theorem.
The former easily shows that the classical bias of a line game is bounded from above by the Gowers $t$-uniformity norm of the game map.
We show that in fact the same upper bound holds for the entangled bias as well.
A large entangled bias thus implies a large uniformity norm for the game map.
Analogous to the above-mentioned octahedral norms for tensors, uniformity norms were introduced to quantify a notion of pseudorandomness for bounded maps over abelian groups as an important step in Gowers' other proof of Szemer\'edi's Theorem, based on higher-order Fourier analysis.
The highly non-trivial Gowers Inverse Theorem of Tao and Ziegler~\cite{TaoZiegler:2012} establishes that high uniformity norm again implies correlation with structure.
Although structure in this context means something quite different than for tensors, we show that it still implies a lower bound on the classical bias.
The above observations together prove Theorem~\ref{thm:line_games}, details of which can be found in Section~\ref{sec:linearFormsGame}.

\subsection{Semidefinite programming relaxation.}
The proof of Theorem~\ref{thm:uniquegames} is a small modification of a proof in~\cite{Charikar:2006}. They consider a semidefinite programming (SDP) relaxation of the optimization problem for the classical value and then give two algorithms for rounding the result of the SDP to a classical strategy. In the SDP relaxation the objective is to optimize $\E_{x,y} \sum_{i=1}^k \langle u^{(x)}_i \mid v^{(y)}_{\pi_{xy}(i)} \rangle$ where $u^{(x)}_i, v^{(y)}_j \in \mathbb{R}^d$ are vectors corresponding to questions $x,y$ and answers $i,j$. Furthermore, $\pi_{xy}$ is the matching of correct answers on questions $x,y$. A classical strategy would correspond to the case where the vectors are integers instead, such that for each $x$ exactly one $u^{(x)}_i$ is equal to 1 and all other $u^{(x)}_i$ are equal to zero and similar for the $v^{(y)}_j$. A quantum strategy also gives rise to a set of vectors, but satisfying different constraints~\cite{Kempe:2008}.
One of the constraints of the SDP considered in~\cite{Charikar:2006} is $0 \leq \langle u_i \mid v_{\pi_{xy}(i)} \rangle \leq |u_i|^2$ which is valid for classical strategies, but in general not for quantum strategies. For our proof, we consider the same SDP but with this constraint dropped. In that case it is also a relaxation for the entangled case and with a few changes one of the rounding algorithms in~\cite{Charikar:2006} is also valid when the constraint is dropped.
Note that the result only beats a trivial strategy when $\epsilon = \mathcal{O}(1 / \log k)$ whereas the other rounding algorithm in~\cite{Charikar:2006} is non-trivial for any $\epsilon$. However this other algorithm is more dependent on the extra constraint and it is not clear if it can be dropped there as well.

To get some intuition for the rounding algorithm, we sketch a solution for $\epsilon=0$ here. In this case one can show that for each question pair $x,y$ the set of vectors $\ket{u^{(x)}_i}$ ($i=1,...,k$) known by the first player is the same set of vectors as the set $\ket{v^{(y)}_i}$ ($i=1,...,k$) known to the second player. In particular, the vector $\ket{u^{(x)}_i}$ is the same as the matching vector $\ket{v^{(y)}_{\pi_{xy}(i)}}$ of the other player. Using shared randomness they can sample a random vector $\ket{g}$ and compute the overlaps $\xi^{(x)}_i = \langle g \ket{u^{(x)}_i}$ and $\xi^{(y)}_i = \langle g \ket{v^{(y)}_i}$ respectively. As they have the same vectors, the players will have the same values for answers in the matching: $\xi^{(x)}_i = \xi^{(y)}_{\pi_{xy}(i)}$. Now both players simply output the answer $i$ for which $|\xi^{(x)}_i|$ (and $|\xi^{(y)}_i|$ for the other player) has the largest value. With probability one this will yield correct answers.
For $\epsilon>0$ the sets of vectors will not be exactly equal and therefore the values $\xi^{(x)}_i,\xi^{(y)}_{\pi_{xy}(i)}$ will be close but not exactly equal. The discrepancy in these values will be bigger for vectors $\ket{u^{(x)}_i}$ with a small norm. In Section~\ref{sec:perfectunique} we provide the rounding algorithm in full detail and show how this issue is solved.

\newpage
\section{Perfect quantum Schmidt strategies for MOD games} \label{sec:perfectghz}

We start by defining a set of games that turn out to characterize the games we are interested in.

\begin{definition}[Angle game] \label{def:anglegame}
    Define an \emph{angle game} as a $t$-player MOD-$m$ game where player $j$ gets an angle $e^{i \phi_j}$ as input, with the promise that $e^{i\phi_1}\cdot...\cdot e^{i\phi_t} = \omega_m^{M(\phi_1,...,\phi_t)}$ where $M(\phi_1,...,\phi_t)\in\{0,1,...,m-1\}$ and $\omega_m = e^{i2\pi/m}$. The players win if and only if the sum of their outputs modulo $m$ is equal to $M(\phi_1,...,\phi_t)$.
\end{definition}
Note that an angle game is completely defined by $t$, $m$ and a probability distribution over angle tuples.
Furthermore, $t$-player Boyer games~\cite{Boyer:2004} with parameters $(D,M)$ are angle games where the (discrete) probability distribution is uniform over all angles of the form $e^{i2\pi x/(MD)}$ with $x=0,1,...,D-1$ whose product is an $M$-th root of unity. The promise $\sum x_j \equiv 0 \mod D$ as stated in the Boyer games translates to $\prod e^{i2\pi x_j / (MD)} = \omega_M^{l}$ in the angle game, where $l=\sum x_j / D$.

\begin{lemma}
    Any angle game has entangled value 1 which can be obtained using a shared GHZ state.
\end{lemma}
\begin{proof}
    Consider the following quantum strategy using a GHZ state of dimension $m$.
    Every player applies the local diagonal unitary $U_{jj} = e^{i \; j\cdot \phi}$ on input $e^{i\phi}$.
    Then every player applies an inverse Fourier transform $F_{ij}^\dagger = \frac{1}{\sqrt{m}} \omega_m^{-i\cdot j}$ and then measures in the computational basis and outputs the result. With probability 1 the sum of their outputs is equal to $l$.
\end{proof}

\begin{lemma} \label{lemma:reducetoangle}
    Any $t$-player MOD-$m$ game with a perfect Schmidt strategy can be reduced to an angle game.
\end{lemma}

\begin{proof}
    Let $\{ P^{(j,x_j)}_1, ..., P^{(j,x_j)}_m \}$ be the projective measurement done by player $j$ on input $x_j$ so that $P^{(j,x_j)}_i$ corresponds to output $i$. This set of projectors is pairwise orthogonal and sums to identity.
    Now define unitaries $U^{(j,x_j)} = \sum_i \omega_m^i P^{(j,x_j)}_i$. Since the strategy is perfect we have for every input $(x_1,...,x_t)$ that
    \begin{align}
        \omega_m^{M(x_1,...,x_t)} &=  \bra{\psi} U^{(1,x_1)} \otimes U^{(2,x_2)} \otimes ... \otimes U^{(t,x_t)} \ket{\psi} \nonumber \\
        &= \sum_{i,j} c_i c_j \bra{e^{(1)}_i} U^{(1,x_1)} \ket{e^{(1)}_j} \; \bra{e^{(2)}_i} U^{(2,x_2)} \ket{e^{(2)}_j} \cdots \bra{e^{(t)}_i} U^{(t,x_t)} \ket{e^{(t)}_j} \nonumber \\
        &= \sum_{i,j} c_i c_j U^{(1,x_1)}_{ij} \; U^{(2,x_2)}_{ij} \cdots U^{(t,x_t)}_{ij} . \label{eq:perfectunitaries}
    \end{align}
    where we entered the definition of a Schmidt state as given in Section~\ref{sec:introperfectghz} and we shortened $U^{(t,x_t)}_{ij} := \bra{e^{(t)}_i}U^{(t,x_t)}\ket{e^{(t)}_j}$.
    Now apply Cauchy-Schwarz to obtain
    \begin{align*}
        \left\vert \sum_{i,j} c_i c_j U^{(1,x_1)}_{ij} \; U^{(2,x_2)}_{ij} \cdots U^{(t,x_t)}_{ij} \right\vert
        &\leq
        \left( \sum_{i,j} c_i^2 \left\vert U^{(1,x_1)}_{ij} \right\vert^2 \right)^{\frac{1}{2}}
        \left( \sum_{i,j} c_j^2 \left\vert U^{(2,x_2)}_{ij} \cdots U^{(t,x_t)}_{ij} \right\vert^2 \right)^{\frac{1}{2}} \\
        &\leq
        1 \; \left( \sum_{i,j} c_j^2 \left\vert U^{(2,x_2)}_{ij} \right\vert^2 \right)^{\frac{1}{2}} = 1 .
    \end{align*}
    Here we used that the $U^{(j,x_j)}$ are unitary and therefore their rows and columns are unit vectors.
    When $|\langle a,b\rangle| = \norm{a} \cdot \norm{b}$ then we have $\ket{a} = \lambda \ket{b}$ for some $\lambda \in \mathbb{C}$. Keeping in mind the complex conjugation in the inner product, there is a $\lambda$ such that
    \begin{align*}
        \lambda c_i \overline{U^{(1,x_1)}_{ij}} = c_j U^{(2,x_2)}_{ij} \cdots U^{(t,x_t)}_{ij}
    \end{align*}
    where $\overline{z}$ denote the complex conjugate of $z$. Plugging this into~\eqref{eq:perfectunitaries} gives $\lambda = \omega_m^{M(x_1,...x_t)}$.
    From the above equation it follows that when $U^{(t,x_t)}_{ij}$ is non-zero for $t=1$ then it is non-zero for every $t$.
    Instead of the first player we could have used any other player in the above derivation, so if any $U^{(t,x_t)}_{ij}$ is non-zero for some $t$ then it is non-zero for all $t$.
    Let $i,j$ be such that $U^{(t,x_t)}_{ij}\neq 0$, then we can take the argument of the above equation to find
    \begin{align*}
        \frac{2\pi}{m} M(x_1,...,x_t) = \arg(U^{(1,x_1)}_{ij}) + \arg(U^{(2,x_2)}_{ij}) + \cdots + \arg(U^{(t,x_t)}_{ij}) .
    \end{align*}
    On any input $(x_1,...,x_t)$, the players simply looks at the first non-zero element of their matrix $U^{(t,x_t)}$ and look at the argument $\phi_t := \arg(U^{(t,x_t)}_{ij})$. These angles have the property that $e^{i\phi_1}\cdot ... \cdot e^{i\phi_t} = \omega_m^{M(x_1,...,x_t)}$. This reduces the game to an angle game.
\end{proof}

\begin{definition}[Connected inputs] \label{def:connected}
    For any game, define a graph where every input (a $t$-tuple) with non-zero probability of being asked is a vertex. Two inputs are connected via an edge if they differ on only one player and agree on the other $t-1$ coordinates.
    We say the game has \emph{connected inputs} if this graph is connected.
\end{definition}

Dinur et al.~\cite{Dinur:2016} consider the same graph which they call the $(t-1)$-connection graph of the game.
Total games and free games have connected inputs.
Games that do \emph{not} have connected imputs typically have a \emph{promise} on the inputs.

\begin{lemma} \label{lemma:connected}
    Any angle game with connected inputs has classical value 1.
\end{lemma}

\begin{proof}
    Fix an input $(e^{i \alpha_1},...,e^{i \alpha_t})$. Now define $\beta_1(e^{i \phi_1}) = e^{i \phi_1} e^{i\alpha_2} e^{i\alpha_3}\cdots e^{i\alpha_t}$ and $\beta_j(e^{i \phi_j}) = e^{i \phi_j} e^{-i \alpha_j}$ for $j\geq 2$. The product of the angles is left unchanged under these maps, $\beta_1(e^{i\phi_1})\beta_2(e^{i\phi_2})\cdots\beta_t(e^{i\phi_t}) = e^{i\phi_1}\cdots e^{i\phi_t}$. We claim that every input is mapped to an $m$-th root of unity, i.e. $\beta_j(e^{i \phi_j}) = \omega_m^{l_j}$. Therefore player $j$ can output $l_j$ and $\sum_j l_j = M(e^{i\phi_1},...,e^{i\phi_t})$ thus winning the game with probability 1.
    First note that on the fixed input we have $\beta_1(e^{i\alpha_1}) = \omega_m^{M(e^{i\alpha_1},...,e^{i\alpha_t})}$ and $\beta_j(e^{i\alpha_j}) = 1$ for $j\geq 2$. so the claim holds on the fixed input.
    By the strong connectivity we can obtain another input to the game by changing only the input for a single player. We now show that when the claim holds for one input then it also holds when only one player's value is changed. Using these single-player edits we can eventually reach all inputs. 
    Let $(e^{i\phi_1},...,e^{i\phi_j},...,e^{i\phi_t})$ and $(e^{i\phi_1},...,e^{i\phi'_j},...,e^{i\phi_t})$ be two inputs that only differ for player $j$. Now assume that the claim holds for the first input. We then have
    \begin{align*}
        \omega_m^{M(e^{i\phi_1},...,e^{i\phi'_j},...,e^{i\phi_t})}
        &= \beta_1(e^{i\phi_1})\cdots\beta_j(e^{i\phi'_j})\cdots\beta_t(e^{i\phi_t}) \\
        &= \beta_1(e^{i\phi_1})\cdots\beta_j(e^{i\phi_j})\cdots\beta_t(e^{i\phi_t}) \cdot \beta_j(e^{i\phi_j})^{-1} \beta_j(e^{i\phi'_j}) \\
        &= \omega_m^{M(e^{i\phi_1},...,e^{i\phi_j},...,e^{i\phi_t})} \omega_m^{-l_j} \beta_j(e^{i\phi'_j})
    \end{align*}
    from which it follows that $\beta_j(e^{i\phi'_j}) = \omega_m^{l'_j}$ for some $l'_j$.
\end{proof}

Lemma~\ref{lemma:perfectghztotal} follows directly from Lemma~\ref{lemma:reducetoangle} and Lemma~\ref{lemma:connected}.

\subsection{Classical strategies for angle games} \label{sec:uag}

Having characterized our class of games as angle games we proceed by presenting classical strategies for these games.
Our aim is to provide strategies that work for any probability distribution on the set of inputs.
In this section it will be convenient to write the angles as $e^{i\frac{2\pi}{m} \phi}$ so that $\phi$ runs from $0$ to $m$ instead of $0$ to $2\pi$.

\begin{definition} \label{def:uag}
    Define the probability distribution $\pi_{U_t}$ on the set
    \begin{align*}
        U_t = \left\{ (e^{i\frac{2\pi}{m}\phi_1},...,e^{i\frac{2\pi}{m}\phi_t}) \mid \phi_j \in [0,1) \;\;,\;\; e^{i\frac{2\pi}{m}\phi_1}\cdots e^{i\frac{2\pi}{m}\phi_t} = \omega_m^{l} ,\; l\in\{0,1,...,m-1\}  \right\} 
    \end{align*}
    as follows: for $1\leq j \leq t-1$, draw $\phi_j$ independently uniformly at random from $[0,1)$. Then define $\phi_t$ as the unique number in $[0,1)$ that makes the product $e^{i\frac{2\pi}{m}\phi_1}\cdots e^{i\frac{2\pi}{m}\phi_t}$ an $m$-th root of unity.
    We define the \emph{uniform angle game}, denoted $\mathrm{UAG}_{t,m}$, as a $t$-player MOD-$m$ angle game (Definition~\ref{def:anglegame}) where the input distribution is $\pi_{U_t}$.
\end{definition}

As stated before, in a $t$-player Boyer game~\cite{Boyer:2004} with parameters $(D,M)$ the (discrete) probability distrubtion is uniform over all angles of the form $e^{i2\pi x/(MD)}$ with $x=0,1,...,D-1$ whose product is an $M$-th root of unity. This is similar to the $\pi_{U_t}$ distribution but where the angles $\phi_j \in [0,1)$ are now discrete $\phi_j \in \{0, \frac{1}{D}, \frac{2}{D}, ..., \frac{D-1}{D} \}$.

The distribution $\pi_{U_t}$ is the hardest distribution as captured by the following claim.
\begin{claim} \label{claim:hardestdistribution}
    Let $G$ be a $t$-player MOD-$m$ angle game with input distribution $\pi_G$. Then $\omega(G) \geq \omega(\mathrm{UAG}_{t,m})$.
\end{claim}
\begin{proof}
    Assume the players get an input $(e^{i\frac{2\pi}{m}\phi_1},...,e^{i\frac{2\pi}{m}\phi_t})$ with $e^{i\frac{2\pi}{m}\phi_1}\cdots e^{i\frac{2\pi}{m}\phi_t} = \omega_m^a$ from $\pi_G$.
    Using shared randomness, draw $t-1$ independent random angles $e^{i\alpha_1},...,e^{i\alpha_{t-1}}$ where each $\alpha_i$ is uniform on $[0,2\pi)$. Multiply the input $e^{i\frac{2\pi}{m}\phi_j}$ of player $j$ by $e^{i\alpha_j}$ for $j \leq t-1$ and multiply $e^{i\frac{2\pi}{m}\phi_t}$ by $e^{-i(\alpha_1+...+\alpha_{t-1})}$ to preserve the product. The resulting distribution is uniform on the set
\begin{align*}
    \left\{ (e^{i\frac{2\pi}{m}\phi_1},...,e^{i\frac{2\pi}{m}\phi_t}) \mid \phi_j \in [0,m) \;\;,\;\; e^{i\frac{2\pi}{m}\phi_1}\cdots e^{i\frac{2\pi}{m}\phi_t} = \omega_m^a \right\}
\end{align*}
One can always write $\phi_j = l_j + \varphi_j$ with $l_j\in\{0,1,...,m-1\}$ and $0\leq\varphi_j<1$.
    Note that $(e^{i\frac{2\pi}{m}\varphi_1},...,e^{i\frac{2\pi}{m}\varphi_t})$ is distributed according to $\pi_{U_t}$ so the players can play a strategy for $\mathrm{UAG}_{t,m}$ to obtain answers $(a_1,...,a_t)$.
    On input $e^{i \frac{2\pi}{m}(l_j + \varphi_j)}$, player $j$ outputs $a_j+l_j$. They are correct if and only if the answers $a_j$ are correct for $\mathrm{UAG}_{t,m}$ on input $(e^{i\frac{2\pi}{m}\varphi_1},...,e^{i\frac{2\pi}{m}\varphi_t})$. This proves the claim.
\end{proof}

\begin{claim} \label{claim:uaglowerbound}
    The uniform angle game satisfies $\omega(\mathrm{UAG}_{t,m}) \geq \frac{1}{m} + \frac{m-1}{m} t^{1-t}$.
\end{claim}
\begin{proof}
    The players get inputs from $U_t$ as in Definition~\ref{def:uag}.
    On input $\phi_i$, player $i$ computes $x_i = \lfloor t \; \phi_i \rfloor$, so that $x_i \in \{0,...,t-1\}$. We then have $\frac{1}{t} \sum_{i=1}^t x_i \leq \sum_{i=1}^t \phi_i < 1 + \frac{1}{t}\sum_{i=1}^t x_i$ and since the correct answer $l$ is given by $l = \sum_{i=1}^t \phi_i$, we see that $l$ is uniquely determined by the sum of $x_i$. Now using shared randomness the players sample $t-1$ random numbers from $\{0,...,t-1\}$. With probability $\frac{1}{t^{t-1}}$ these numbers are exactly equal to $x_1,...,x_{t-1}$. The last player assumes that the random numbers are indeed $x_1,...,x_{t-1}$ and outputs the correct $l$. The other players output $0$ if their $x_i$ matches the random sample and output a random number otherwise. This yields a bias of $\frac{1}{t^{t-1}}$, or winning probability of $\frac{1}{m} + \frac{m-1}{m} t^{1-t}$.
\end{proof}

Theorem~\ref{thm:perfectghz} follows from Lemma~\ref{lemma:reducetoangle}, Claim~\ref{claim:hardestdistribution} and Claim~\ref{claim:uaglowerbound}.

We proceed by describing a strategy for $\mathrm{UAG}_{t,m}$ that improves on the bound of Claim~\ref{claim:uaglowerbound}.
Let $(\phi_1,...,\phi_t)$ be drawn from $\pi_{U_t}$. Define $\Phi = \phi_1+\phi_2+...+\phi_{t-1}$ then by definition of $\pi_{U_t}$, $\Phi$ is the sum of $t-1$ independent uniform $[0,1)$ variables. Furthermore, the last input is $\phi_t = \lceil\Phi\rceil - \Phi$ and the correct answer $l$ is defined by $l \equiv \lceil\Phi\rceil \mod m$.
The distribution of $\Phi$ is known as the Irwin-Hall distribution \cite{Irwin:1927,Hall:1927}:
\begin{align*}
    \mathbb{P}(\Phi \leq x) = \frac{1}{(t-1)!} \sum_{j=0}^{\lfloor x \rfloor} (-1)^j \binom{t-1}{j} (x-j)^{t-1} .
\end{align*}
We consider a set of strategies that we call the \emph{semi-trivial strategies}, in which the first $t-1$ players always output 0. The last player then plays optimally when given the input $x$.
We conjecture that this strategy is optimal for the uniform angle game. In the semi-trivial strategy, the last player chooses the $l$ that maximizes
\begin{align} \label{eq:marginalprobability}
    \mathbb{P}(\text{correct answer is }l \mid \text{last input is } x) =
    \mathbb{P}( \lceil \Phi \rceil \equiv l \mod m \mid \phi_t = x ) .
\end{align}
This probability is plotted as a function of $x$ in Figure~\ref{fig:marginalprobability}.
\begin{figure}
    \centering
    \includegraphics[width=\columnwidth]{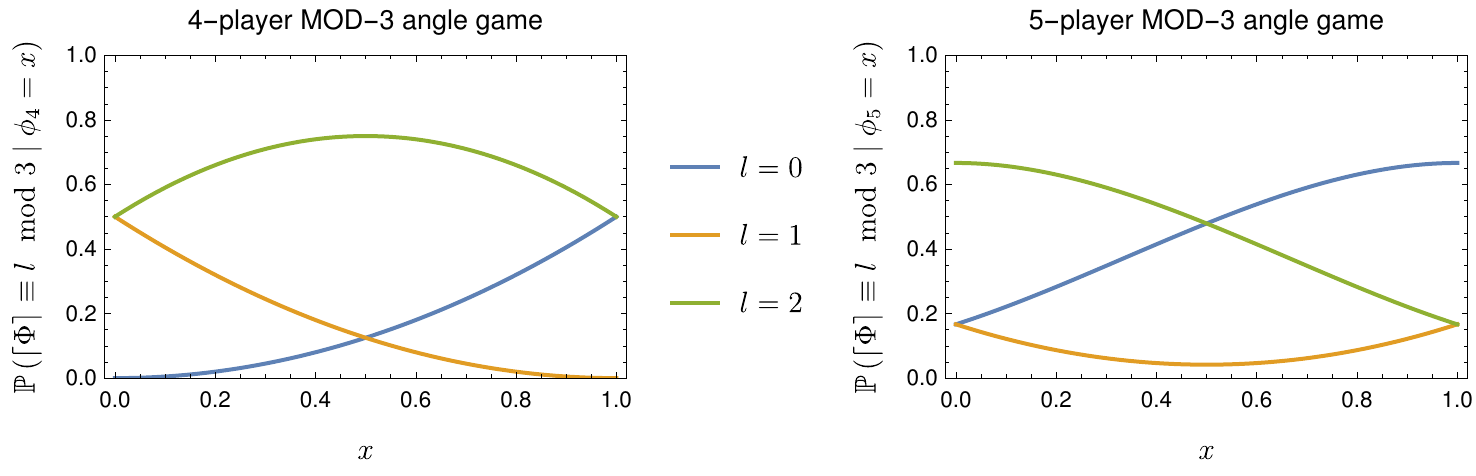}
    \caption{\label{fig:marginalprobability} Probability that the correct answer of the angle game is $l$, conditioned on the last player receiving input $x$, as defined in Equation~\eqref{eq:marginalprobability}.}
\end{figure}
The figure shows that for a 4-player MOD-3 game (left plot) the optimal choice for the last player is to ignore the input $x$ and always output $2$. Interestingly, this is the pattern we observe for any value of $m$ for any even number of players (checked up to $m=10,t=10$): the value of $l$ for which this probability is maximal is independent of $x$. This means that these strategies (for $t$ even) are locally optimal in the sense that changing any single player's strategy will not improve the winning probability.
The right plot of the figure shows that for 5 players the optimal answer depends on whether $x\leq \frac{1}{2}$ or $x > \frac{1}{2}$. This, too, is seems to be a general pattern for any $m$ and any odd number of players (checked up to $m=10,t=11$).
The winning probabilities provided by these semi-trivial strategies are given in Table~\ref{tab:winningprobs} for $m=2$ and $m=3$.
The semi-trivial strategies are lower bounds for the winning probability of \emph{every} $t$-player angle game. We can find upper bounds by finding upper bounds for particular angle games. The table provides some upper bounds obtained from brute-force searching through all strategies for Boyer games.

\begin{table}
    \begin{center}
	\resizebox{\columnwidth}{!}{%
        \begin{tabular}{|l|c|c|c|c|c|c|c|c|}
            \hline
            $t$ (\# players) & 2 & 3 & 4 & 5 & 6 & 7 & 8 & 9 \\
            \hline
            \multicolumn{9}{|l|}{ MOD 2 } \\
            \hline
            lower bound & 1 & 3/4 & 2/3 & 29/48 & 17/30 & 781/1440 & 166/315 & 8341/16128 \\
                        & 1 & 0.75 & 0.6667 & 0.6042 & 0.5667 & 0.5424 & 0.5270 & 0.5172 \\
            \hline
            upper bound & 1 & 3/4  & 43/64  & 155/256 & 583/1024 & 35/64  & 273/512 & 1056/1048\\
                        & 1 & 0.75 & 0.6719 & 0.6055  & 0.5693   & 0.5469 & 0.5332  & 0.5200 \\
            \hline
            \multicolumn{9}{|l|}{ MOD 3 } \\
            \hline
            lower bound & 1 & 3/4 & 2/3 & 115/192 & 11/20 & 785/1536 & 403/840 & 260451/573440 \\
                        & 1 & 0.75 & 0.6667 & 0.5990 & 0.5500 & 0.5111 & 0.4798 & 0.4542 \\

            \hline
            upper bound & 1 & 61/81  & 163/243  & 17/27  & 47/81  & 131/243 & 41/81  & 349/729 \\
                        & 1 & 0.7531 & 0.6708 & 0.6296 & 0.5802 & 0.5391  & 0.5062 & 0.4787 \\
            \hline
        \end{tabular}
    }
    \end{center}
    \caption{\label{tab:winningprobs} Lower and upper bounds for the winning probabilities of the $\mathrm{UAG}_{t,m}$ games (Definition~\ref{def:uag}). The lower bounds are the semi-trivial strategies described in the text. The upper bounds were found by iterating through all strategies for $t$-player Boyer games with values of $D$ up to $D=9$ when computation time allowed it.}
\end{table}

For the case of 3-player XOR games, the upper bound $\omega(\mathrm{UAG}_{3,2}) = 3/4$ is tight. For the 4-player case we have $\omega(\mathrm{UAG}_{4,2}) \geq 2/3$ and it seems that searching through Boyer games gives increasingly better bounds, approaching $2/3$ as the input size $D$ is increased. However, one can show that for any finite $D$ the lower bound of $2/3$ will not be reached, because when $D=2^m$ there is a strategy that achieves a winning probability of $\frac{2}{3}+\frac{4^{-m}}{3}$. The corresponding strategy is that the first $3$ players output $0$ and the fourth player outputs $1$ when their input is $0$ or $1$ and outputs $0$ when their input is $2,...,D-1$. This strategy is optimal for $D=2,4,8$. One can show that for values of $D$ that are not a power of $2$, i.e. $D=D' 2^m$ the game reduces to one with $D=2^m$.

\section{Free XOR games} \label{sec:hypergraphs}
In this section we will define free XOR games and give the definition of hypergraph norms (only for real-valued functions on discrete domains). For more details on hypergraph norms we refer to \cite{HatamiPhD}. We will then relate the hypergraph norm, with respect to a certain hypergraph, of the game tensor to the quantum bias of free XOR games. Our main tool is a Cauchy-Schwarz type of inequality for operators, that is why we will state it here.

\begin{proposition} \label{prop:opcs}
Let $A_i,B_i\in \mathrm{End}(\C^n)$ for $i=1,\dots,k$. Then
\begin{align*}
\Vert \sum_{i\in [k]}A_iB_i\Vert\leq \Vert\sum_{i\in [k]}A_iA_i^*\Vert^{1/2}\;\Vert\sum_{i\in [k]}B_i^*B_i\Vert^{1/2},
\end{align*}
where all the norms are operator norms.
\end{proposition}
\begin{proof}
Write  
\begin{align*}
A=\begin{bmatrix}
A_1&A_2&\cdots & A_k\\
0&0&\cdots& 0\\
\vdots&\vdots&\vdots&\vdots\\
0&0&\cdots& 0
\end{bmatrix}
\text{and }
B=\begin{bmatrix}
B_1&0&\cdots & 0\\
B_2&0&\cdots& 0\\
\vdots&\vdots&\vdots&\vdots\\
B_k&0&\cdots& 0
\end{bmatrix}.
\end{align*}
Then we use the fact that
\begin{align*}
AB=\begin{bmatrix}
\sum_{i\in[k]}A_iB_i&0&\cdots & 0\\
0&0&\cdots& 0\\
\vdots&\vdots&\vdots&\vdots\\
0&0&\cdots& 0
\end{bmatrix},
\end{align*}
together with the properties
\begin{align*}
&\Vert C\otimes E_{ij}\Vert=\Vert C\Vert,\\
&\Vert C\Vert^2=\Vert C^* C\Vert=\Vert CC^*\Vert,
\end{align*} 
to conclude
\begin{align*}
&\Vert \sum_{i\in[k]}A_iB_i\Vert=\Vert \sum_{i\in[k]}A_iB_i\otimes E_{11}\Vert=\Vert AB\Vert\leq \Vert A\Vert\Vert B\Vert\\
&=\Vert AA^*\Vert^{1/2}\Vert B^*B\Vert^{1/2}=\Vert\sum_{i\in [k]}A_iA_i^*\Vert^{1/2}\Vert\sum_{i\in [k]}B_i^*B_i\Vert^{1/2}.
\end{align*}
\end{proof}
A $t$-player free XOR game $G$ is given by finite non-empty sets $X_1,\dots,X_t$, a product distribution over $X:=X_1\times\cdots\times X_t$ and a game tensor
\begin{equation}
T\colon X\rightarrow \{\pm 1\}. 
\end{equation}
The classical bias of the free XOR game $G$, which we denote by $\beta(G)$ is given by
\begin{align*}
\beta(G):=\max_{a_i\colon X_i\rightarrow \{\pm 1\}}\vert\E_{(x_1,\dots,x_t)\in X}T(x_1,\dots,x_t)\prod_{i=1}^ta_i(x_i)\vert.
\end{align*}
The quantum bias of the free XOR game $G$, which we denote by $\beta^*(G)$ is given by the expression
\begin{equation}
\beta^*(G):=\max_{N\in \N, A_i\colon X_i\rightarrow \mathrm{Obs}^{\pm}(\C^N)}\Vert\E_{(x_1,\dots,x_t)\in X}T(x_1,\dots,x_t)\prod_{i=1}^tA_i(x_i)\Vert_{\text{op}},
\end{equation}
the maximization is taken over $\pm$-observable valued functions $A_i$ such that $[A_i,A_j]=0$ for $i\neq j$, this corresponds to a quantum strategy of the players. The expectation is taken over the given distribution.

Before we go into the detail of the proof of Theorem~\ref{thm:xor_biases} for any number of players, we first sketch the core idea of the proof for \emph{two} players, for which we do not yet need to resort to hypergraphs. For a two-player game $G$ with game tensor $T$, the commuting-operator strategies $A,B$ yield a bias of
\begin{align*}
    \eta=\Vert\E_{(x,y)\in X\times Y}T(x,y)A(x)B(y)\Vert.
\end{align*}
where the norm is the operator norm. Using Proposition~\ref{prop:opcs} we peel off the operator $B(y)$
\begin{align*}
    \eta &= \Big\Vert \E_{y\in Y} \Big( \E_{x\in X}T(x,y)A(x) \Big) B(y) \Big\Vert.\tag{independent questions}\\
         &\leq \Big\Vert \E_{y\in Y} \Big( \E_{x\in X}T(x,y)A(x) \Big) \Big( \E_{x'\in X}T(x',y)A(x') \Big)^* \Big\Vert^{1/2} \;\;\; \Big\Vert \E_{y\in Y} B(y) B(y)^* \Big\Vert^{1/2} \\
         &\leq \Big\Vert \E_{y\in Y} \E_{x,x'\in X} T(x,y) T(x',y) A(x) A(x')^* \Big\Vert^{1/2} \tag{using $\Vert B(y) \Vert \leq 1$}.
\end{align*}
Now we apply the inequality again on the sum over $(x,x')$ to get rid of the $A$ operator.
\begin{align*}
    \eta &\leq \Big\Vert \E_{x,x' \in X} \Big( \E_{y\in Y} T(x,y) T(x',y) \Big) A(x) A(x')^* \Big\Vert^{1/2} \\
         &\leq \Big\vert \E_{x,x'} \big( \E_{y} T(x,y) T(x',y) \big) \big( \E_{y'} T(x,y') T(x',y') \big) \Big\vert^{\frac{1}{4}} \;
         \Big\Vert \E_{x,x'} \big( A(x) A(x')^* \big) \big(A(x) A(x')^* \big)^* \Big\Vert^{\frac{1}{4}} \\
         &\leq \Big\vert \E_{x,x'\in X} \E_{y,y'\in Y} T(x,y) T(x',y) T(x,y') T(x',y') \Big\vert^{1/4} . \tag{using $\Vert A(x) \Vert \leq 1$}
\end{align*}
We proceed by rewriting the last expression
\begin{align*}
    \eta^4 &\leq \Big\vert \E_{(x',y')\in X\times Y} T(x',y') \E_{(x,y)\in X \times Y} T(x,y) T(x',y) T(x,y') \Big\vert \\
           &\leq \E_{(x',y')\in X\times Y} \Big\vert \E_{(x,y)\in X \times Y} T(x,y) T(x',y) T(x,y') \Big\vert . \tag{triangle inequality}
\end{align*}
By the pigeonhole principle there must be choices of $x',y'$ such that
\begin{align*}
    \eta^4 &\leq \Big\vert \E_{(x,y)\in X \times Y} T(x,y) T(x',y) T(x,y') \Big\vert ,
\end{align*}
which is the expression for the bias of the classical strategies $a(x) = T(x,y')$ and $b(y) = T(x',y)$, proving Theorem~\ref{thm:xor_biases} for $t=2$ players.
For $t\geq 3$ we can apply the same idea, peeling off the operators one by one, but the final expression is more involved. We will now develop the techniques to deal with this. In particular, we need the notion of hypergraph norms. For our purposes, we only consider $t$-uniform hypergraphs which are also $t$-partite. 
\begin{definition}\label{def:hypergraph}
For $t\geq 2$, let $V_1,\dots,V_t$ be finite non-empty sets and $V:=V_1\times\cdots\times V_t$. Given a subset $E\subset V$, we say that the pair $H=(V_1\cup\dots\cup V_t,E)$ is a $t$-partite $t$-uniform hypergraph with vertex set $V_1\cup\dots\cup V_t$ and edge set $E$.
\end{definition}
\begin{definition}\label{def:hypergraph_norm}
Let $t\geq 2$ and $X_1,\dots,X_t$ be finite non-empty sets and suppose a product distribution on $X:=X_1\times\cdots\times X_t$ is given to us. Let $T\colon X\rightarrow \R$ be a function and $H=(V_1\cup\dots\cup V_t,E)$ be a $t$-partite $t$-uniform hypergraph. We define a non-negative function $\Vert\cdot \Vert_H$ on the function $T$ by
\begin{equation}
\Vert T\Vert_H:=\left|\E_{\phi_i\colon V_i\rightarrow X_i}\prod_{(v_1,\dots,v_t)\in E}T(\phi_1(v_1),\dots,\phi_t(v_t))\right|^{\frac{1}{\vert E\vert}}.
\end{equation}
The expectation is taken with respect to the following distribution: a particular map $\phi_i\colon V_i\rightarrow X_i$ occurs with probability $\prod_{v\in V_i}p_i(\phi_i(v))$ where $p_i$ is the probability distribution on $X_i$.
\end{definition}
The particular hypergraph which arises naturally when we study the quantum bias of free XOR games is constructed as follows. Starting with a $t$-partite $t$-uniform hypergraph $H$, write $\db_i(H)$ for the $t$-partite $t$-uniform hypergraph obtained by making two vertex-disjoint copies of $H$ and gluing them together so that the vertices in the two copies of $V_i$ are identified. To construct our hypergraph, we start with the hypergraph given by a single edge $e=(v_1,\dots,v_t)$ and vertex sets of size 1 and apply the doubling operation to all parts, i.e. $\db_1(\db_{2}(\dots \db_t(e) ))$. We denote this hypergraph by $H(t)$. A more useful way to define $H(t)$ is as follows. We will do this first for $t=2$ and explain how to do it for any $t$ afterwards. We use 2-bit strings to identify vertices. We start with the hypergraph with a single edge $(x_{00},y_{00})\in V_1\times V_2$. As we will start using the doubling operator, we make copies of the vertex sets. We can use a table to visualize it.
	
\begin{center}
\begin{tabular}{ c|c|c}
\text{ }& $V_1$&$V_2$\\ \hline
\text{starting position} 	& $x_{00}$ & $y_{00}$\\ \hline
\text{$\db_2$} 	& $x_{01}$ & $y_{00}$\\ \hline
\text{$\db_1$} 	& $x_{00}$ & $y_{10}$\\
\text{ } 					& $x_{01}$ & $y_{10}$ \\ 
\end{tabular}
\end{center}
The table may be read as follows; the rows are the edges of the hypergraph and columns are the vertex sets. In this example we have that $V_1=\{x_{00},x_{01}\}$ and $V_2=\{y_{00},y_{10}\}$ and the edge set consists of $\{(x_{00},y_{00}),(x_{01},y_{00}), (x_{00},y_{10}),(x_{01},y_{10})\}$. The algorithm for constructing the table is as follows: we start with the starting position row, which corresponds to the (hyper)graph with a single edge $(x_{00},y_{00})$, and as we apply the doubling operator $\db_2$, we add a new row (which corresponds to making a vertex-disjoint copy) where we increase the 2nd bit in the subscript of $x$ but leave $y$ alone (so we have a new copy of $V_1$ but not of $V_2$). After this first step we have a graph with vertex sets $V_1=\{x_{00},x_{01}\}$ and $V_2=\{y_{00}\}$ and edge set $\{(x_{00},y_{00}),(x_{01},y_{00})\}$. Next we apply $\db_1$ and we get a new copy of $V_2$, but leave $V_1$ alone. \newline
For arbitrary $t\geq 2$; let $v^i_\omega$ be a formal variable with $i\in [t]$ and $\omega$ is a $t$-bit string. We define for $j\in [t]$ an operation on the formal variable by 
\begin{align*}
\Delta_j(v^i_\omega)&:=v^i_{\omega_1,\dots,\omega_j+1,\dots,\omega_t} \text{ for }j\neq i\\
\Delta_i(v^i_\omega)&:=v^i_\omega.
\end{align*} 
where we add modulo 2. The table then looks like
\begin{center}
\begin{tabular}{ c|c|c|c|c}
\text{ }& $V_1$&$V_2$ &\dots & $V_t$ \\ \hline
\text{starting position} 	& $v^1_{0^t}$ & $v^2_{0^t}$ & \dots & $v^t_{0^t}$\\ \hline
\text{$\db_t$} 	& $\Delta_t(v^1_{0^t})$ & $\Delta_t(v^2_{0^t})$ &\dots & $\Delta_t(v^t_{0^t})$\\ \hline
\text{$\db_{t-1}$} 	& $\Delta_{t-1}(v^1_{0^t})$ & $\Delta_{t-1}(v^2_{0^t})$ & \dots & $\Delta_{t-1}(v^t_{0^t})$\\
\text{ } 					& $\Delta_{t-1}(\Delta_t(v^1_{0^t}))$ & $\Delta_{t-1}(\Delta_t(v^2_{0^t}))$  &\dots & $\Delta_{t-1}(\Delta_t(v^t_{0^t}))$\\ \hline
\dots &\dots &\dots&\dots&\dots\\
\end{tabular}
\end{center}
At step $k$, the algorithm takes all the rows of the previous steps together and applies $\Delta_{t-k+1}$ on each of the formal variables in the rows. We also write $\db_i(e)$ for the row where we apply $\Delta_i$ on each variable of the row $e$. We see in this way that, for example, the edge set of $H(t)$ has cardinality $2^t$ and the number of vertices in each $V_i$ is $2^{t-1}$. In the following proposition we list some properties of $H(t)$ which we prove using this description. We will be using the terms row and edge interchangeably as they mean the same in this context.
\begin{proposition}\label{prop:properties_H(T)}
The hypergraph $H(t)$ has the following properties: (1) it is $t$-partite and $t$-uniform, (2) it is 2-regular and (3) for all vertices $v$ the following holds: let $e,e'$ be the unique edges such that $v\in e$, $v\in e'$ and $e\neq e'$. For $w\in e\setminus \{v\}$, denote by $e,e''$ the unique edges such that $w\in e$, $w\in e''$ and $e \neq e''$. Then $e'\cap e''=\emptyset$.
\end{proposition}
\begin{proof}
(1) follows directly from the algorithm described above using the table. We can prove (2) as follows. Suppose in column $V_i$ we have a vertex in some row/edge which we denote by $v^i_{\omega}$, here $\omega$ is a $t$-bit string. First we note that applying $\db_j$ with $j\neq i$ will change $\omega$ as it will flip the $j$-th bit. There are two cases; either we have already applied $\db_i$ in which case $v^i_\omega$ appears in exactly one more row above the current row, or we have not applied $\db_i$ yet in which case there is no $v^i_\omega$ in an earlier row. It will appear exactly once in a later row since applying $\db_i$ will not change $\omega$. For (3), choose again some vertex $v^i_\omega$ in $V_i$ and denote by $e$ the row which appears first in the table containing $v^i_\omega$. The other row/edge which contains $v^i_\omega$ is $e':=\db_i(e)$. Now, let $v^j_\tau$ be a vertex in $V_j$ with $j\neq i$ and $v^j_\tau\in e$, i.e. it is in the same row as $v^i_\omega$. There are two cases; either $j>i$ in which case $e=\db_j(e'')$ where $e''$ is the other (unique) edge containing $v^j_\tau$. Or $j<i$ and the other edge which contains $v^j_\tau$ is $e'':=\db_j(e)$. In any case, a moments thought shows that $e'\cap e''=\emptyset$.
\end{proof}
The next ingredient is the following lemma. 
\begin{lemma}\label{lemma:upperbound_qbias_freeXOR}
For a $t$-player free XOR game $G$ with game tensor $T$, we have that
\begin{align*}
\beta^*(G)\leq \Vert T\Vert_{H(t)}.
\end{align*}
\end{lemma}
\begin{proof}
For convenience, we write the hypergraph in a slightly different way. Write $\phi_i\colon V_i\rightarrow X_i$ and we define an operation $\Delta_j$ on such maps in the same way as above, i.e.
\begin{align*}
(\Delta_j\phi_i)(v^i_\omega)&=\phi_i(\Delta_jv^i_\omega)\text{ for }j\neq i\\
(\Delta_i\phi_i)(v^i_\omega)&=\phi_i(v^i_\omega).
\end{align*}
Also, using the same symbol, we define on functions $T\colon X_1\times\cdots\times X_t\rightarrow \C$
\begin{align*}
    \Delta_jT(\phi_1(v^1_{\omega^1}),\dots,\phi_t(v^t_{\omega^t})):=T(\phi_1(v^1_{\omega^1}),\dots,\phi_t(v^t_{\omega^t})) T^*((\Delta_j\phi_1)(v^1_{\omega^1}),\dots,(\Delta_j \phi_t)(v^t_{\omega^t})),
\end{align*}
one could think of this operation as a kind of multiplicative derivative. If $T$ were an operator-valued map, we still define it in this way. It is then not hard to see that
\begin{align*}
\Delta_1\dots\Delta_tT(\phi_1(v^1_{0^t}),\dots,\phi_t(v^t_{0^t}))=\prod_{(v^1_{\omega^1},\dots,v^t_{\omega^t})\in E(H(t))}T(\phi_1(\omega^1),\dots,\phi_t(\omega^t)),
\end{align*}
using the table as a description of $H(t)$. So we can write
\begin{align*}
\Vert T\Vert_{H(t)}=\vert\E \Delta_1\dots\Delta_tT(\phi_1(v^1_{0^t}),\dots,\phi_t(v^t_{0^t})) \vert^{1/|E|},
\end{align*}
where the expectation is taken over all maps $\phi_i\colon V_i\rightarrow X_i$ with the particular distribution given in Definition~\ref{def:hypergraph_norm}.\newline
Now let us look at the bias of a two player game $G$ with game tensor $T$ and strategies $A,B$
\begin{align*}
\eta=\Vert\E_{(x,y)\in X\times Y}T(x,y)A(x)B(y)\Vert_{\text{op}}.
\end{align*}
    We will do the example of two players to clarify the idea and will prove it in general afterwards. First we use Proposition~\ref{prop:opcs}, the Cauchy-Schwarz inequality, to peel off strategy $B$
\begin{align*}
\eta &=\Vert\E_{y\in Y}\big(\E_{x\in X}T(x,y)A(x)\big)B(y)\Vert_{\text{op}}\\
&\leq \Vert \E_y\big(\E_x T(x,y)A(x)\big)\big(\E_x T(x,y)A(x)\big)^*\Vert_{\text{op}}^{1/2} \;\; \Vert\E_yB(y)^*B(y)\Vert_{\text{op}}^{1/2}\\
&\leq \Vert \E_{y,x,x'}T(x,y)T(x',y)A(x)A(x')^* \Vert_{\text{op}}^{1/2}.
\end{align*}
In the second inequality we used that operator norm of strategies are smaller or equal to 1. We will use the Cauchy-Schwarz inequality one more time, now to peel off strategy $A$
\begin{align*}
\eta & \leq \Vert \E_{x,x'}(\E_y T(x,y)T(x',y))A(x)A(x')^* \Vert_{\mathrm{op}}^{1/2}\\
&\leq \Vert \E_{x,x'} (\E_y T(x,y)T(x',y))(\E_y T(x,y)T(x',y))^* \Vert_{\mathrm{op}}^{1/4} \;\; \Vert\E_{x,x'}(A(x)A(x')^*)^*A(x)A(x')^*\Vert_{\mathrm{op}}^{1/4}\\
& \leq \vert \E_{x,x',y,y'} T(x,y)T(x',y)T(x,y')T(x',y') \vert^{1/4}.
\end{align*}
    Now, to see that this last expression is equal to $\Vert T\Vert_{H(2)}$, we write the expectation in a a different way. Instead of writing $\E_{x,x'}$ we write $\E_{\phi\colon V\rightarrow X}$ where $V=\{v_0,v_1\}$ is a vertex set, so that $x=\phi(v_0)$ and $x' = \phi(v_1)$. Similarly, instead of $\E_{y,y'}$ we write $\E_{\psi\colon W\rightarrow Y}$ where $W=\{w_0,w_1\}$ and we view $H(2)$ to be on this vertex sets. Then, we evaluate $T$ on the edges of $H(2)$, so
\begin{align*}
\vert \E_{x,x',y,y'} T(x,y)T(x',y)T(x,y')T(x',y') \vert^{1/4}=\vert \E_{\phi\colon V\rightarrow X,\psi\colon W\rightarrow Y} \prod_{(v,w)\in E(H(2))}T(\phi(v),\psi(w)) \vert^{1/4}.
\end{align*}

In general, for $t$ players, the proof is as follows
\begin{align*}
\eta:=& \Vert\E T(\phi_1(v^1_{0^t}),\dots,\phi_t(v^t_{0^t}))\left(\prod_{i\in [t-1]}A_i(\phi_i(v^i_{0^t}))\right)A_t(\phi_t(v^t_{0^t}))\Vert_{\mathrm{op}}\\
&\leq \Vert\E T(\phi_1(v^1_{0^t}),\dots,\phi_t(v^t_{0^t}))T(\phi_1(v^1_{0^{t-1}1}),\dots,\phi_t(v^t_{0^t})) \prod_{i\in [t-1]}A_i(\phi_i(v^i_{0^t}))A_i(\phi_i(v^i_{0^{t-1}1}))^*\Vert_{\mathrm{op}}^{1/2}\\
&=\Vert\E\Delta_tT(\phi_1(v^1_{0^t}),\dots,\phi_t(v^t_{0^t}))\prod_{i\in [t-1]}\Delta_t A_i(\phi_i(v^i_{0^t}))\Vert_{\mathrm{op}}^{1/2}.
\end{align*}
Now assume that we have applied the Cauchy-Schwarz inequality $1<n<t$ times to peel off the last $n$ operators and we have obtained the expression
\begin{align*}
\eta\leq \Vert\E\Delta_{t-n+1}\cdots\Delta_tT(\phi_1(v^1_{0^t}),\dots,\phi_t(v^t_{0^t}))\prod_{i\in [t-n]}\Delta_{t-n+1}\cdots\Delta_t A_i(\phi_i(v^i_{0^t}))\Vert_{\mathrm{op}}^{1/2^n}.
\end{align*}
Now apply Cauchy-Schwarz inequality to remove the operator $\Delta_{t-n+1}\cdots\Delta_t A_{t-n}(\phi_i(v^i_{0^t}))$ so that we obtain
\begin{align*}
\eta\leq \Vert\E\Delta_{t-n}\cdots\Delta_tT(\phi_1(v^1_{0^t}),\dots,\phi_t(v^t_{0^t}))\prod_{i\in [t-n-1]}\Delta_{t-n}\cdots\Delta_t A_i(\phi_i(v^i_{0^t}))\Vert_{\mathrm{op}}^{1/2^{n+1}}.
\end{align*}
This completes the induction. Putting $n=t-1$ we have the inequality
\begin{align*}
\eta\leq \vert\E\Delta_1\cdots\Delta_tT(\phi_1(v^1_{0^t}),\dots,\phi_t(v^t_{0^t}))\vert^{1/2^t}.
\end{align*}
\end{proof}
We are now ready to give a proof of Theorem~\ref{thm:xor_biases}
\begin{proof}[Proof of Theorem~\ref{thm:xor_biases}]
We assume $\beta^*(G)>\eta$. Lemma~\ref{lemma:upperbound_qbias_freeXOR} immediately implies $\Vert T\Vert_{H(t)}>\eta$. To construct a classical strategy, we choose an edge $e^*=(v_1^*,\dots,v_t^*)\in E(H(t))$. Any choice of edge works for our argument. $H(t)$ is 2-regular (by Proposition~\ref{prop:properties_H(T)}), so denote by $e_i^*$ the unique edge different from $e^*$ such that $v_i^*\in e_i^*$. Write $e_i^*=(v_1^{(i)},\dots,v_i^*,\dots,v_t^{(i)})$ and $V_i':=V_i\setminus\{v_i^*\}$. Using Proposition~\ref{prop:properties_H(T)} we see that $v_j^*\notin e_i^*$ whenever $i\neq j$. Then
\begin{align*}
\eta^{2^t} & <\vert\E_{\phi_i\colon V_i\rightarrow X_i}\prod_{(v_1,\dots,v_t)\in E}T(\phi_1(v_1),\dots,\phi_t(v_t))\vert\\
&=\vert \E_{\phi_i\colon V_i'\rightarrow X_i} [ \prod_{(v_1,\dots,v_t)\in E\setminus\{e^*,e_1^*,\dots,e_t^*\}}T(\phi_1(v_1),\dots,\phi_t(v_t))\\
&\E_{\phi_i^*\colon\{v_i^*\}\rightarrow X_i} T(\phi_1^*(v_1^*),\dots,\phi_t(v_t^*)) T(\phi_1^*(v_1^*),\dots,\phi_t(v_t^{(1)}))\cdots T(\phi_1(v_1^{(t)}),\dots,\phi_t^*(v_t^*)) ] \vert\\
&\leq \E_{\phi_i\colon V_i'\rightarrow X_i}\vert \E_{\phi_i^*\colon\{v_i^*\}\rightarrow X_i}T(\phi_1^*(v_1^*),\dots,\phi_t^*(v_t^*))T(\phi_1^*(v_1^*),\dots,\phi_t(v_t^{(1)}))\\
& \cdots T(\phi_1(v_1^{(t)}),\dots,\phi_t^*(v_t^*))\vert.
\end{align*}
Let us explain the second and third line in detail. Write $V_i=V_i'\cup \{v_i^*\}$. Any map $\phi_i\colon V_i\rightarrow X_i$ can be given by two maps $\phi_i'\colon V_i'\rightarrow X_i$ and $\phi_i^*\colon \{v_i^*\}\rightarrow X_i$ by defining $\phi_i(v)$ to be $\phi'_i(v)$ when $v\in V_i'$ and otherwise equal to $\phi_i^*(v)$. It can then be seen that
\begin{align*} 
 \E_{\phi_i\colon V_i\rightarrow X_i}(\text{some expression})=\E_{\phi_i'\colon V_i'\rightarrow X_i}[\E_{\phi_i^*\colon \{v_i^*\}\rightarrow X_i}(\text{some expression})].
 \end{align*}
After this we use the triangle inequality. Using the pigeonhole principle we see that there exists specific choices of maps $\phi_i\colon V_i'\rightarrow X_i$ such that
\begin{align*}
\vert \E_{\phi_i^*\colon\{v_i^*\}\rightarrow X_i}T(\phi_1^*(v_1^*),\dots,\phi_t^*(v_t^*))T(\phi_1^*(v_1^*),\dots,\phi_t(v_t^{(1)}))\cdots T(\phi_1(v_1^{(t)}),\dots,\phi_t^*(v_t^*))\vert>\eta^{2^t}.
\end{align*}
The expectation over $t$-tuples of maps $\phi_i^*\colon \{v_i^*\}\rightarrow X_i$ is the same as the expectation over $t$-tuples $x_i^*\in X_i$ and by defining $$a_i(x_i^*):=T(\phi_1(v_1^{(i)}),\dots,x_i^*,\dots,\phi_t(v_t^{(i)}))$$
we see that 
\begin{align*}
\vert \E_{x_1^*,\dots,x_k^*}T(x_1^*,\dots,x_k^*)\prod_{i=1}^ka_i(x_i^*)\vert>\eta^{2^t},
\end{align*}
in other words, the classical bias is at least $\eta^{2^t}$.
\end{proof}

\section{Linear forms game}\label{sec:linearFormsGame}
Before we will go in to the details of the proof of Theorem~\ref{thm:line_games}, we briefly discuss some concepts from higher order Fourier analysis. The reference for this subsection is \cite{tao2012higher}.
 
\subsection{Preliminaries}
Let $G$ be a finite abelian group and $f\colon G\rightarrow \C$ a complex-valued function on $G$. First we define the multiplicative derivative $\Delta_h$ for any $h\in G$ for such functions
\begin{align*}
\Delta_h f(x):=f(x+h)\overline{f(x)}.
\end{align*}
We define for any $s\geq 1$ the Gowers norm $\Vert\cdot\Vert_{U^{s}(G)}$
\begin{align}\label{def:gowersNorm}
\Vert f\Vert_{U^{s}(G)}:=(\E_{h_1,\dots,h_s,x\in G}\Delta_{h_1}\cdots\Delta_{h_s}f(x))^{1/2^s}.
\end{align}
For $s=1$ we get the absolute value of the mean of the function 
\begin{align*}
    \Vert f\Vert_{U^{1}(G)}:=(\E_{h,x\in G}\Delta_{h}f(x))^{1/2} = \vert\E_{x\in G} f(x) \vert,
\end{align*}
so technically it is not a norm, but for $s>1$ it is indeed a norm. By the recursion
\begin{align*}
\Vert f\Vert_{U^{s+1}(G)}^{2^{s+1}}=\E_{h\in G}\Vert\Delta_h f\Vert_{U^s(G)}^{2^s}
\end{align*}
one sees that the expectation in Equation~\ref{def:gowersNorm} is a non-negative real. \\
Now let $\psi_0,\dots,\psi_t\colon G^d\rightarrow G$ be affine linear forms, i.e. maps of the form $\psi_i(g_1,\dots,g_m)=c_i+\sum_{j=1}^mc_{ij}g_j$ wher $c_i\in G$ and $c_{ij}\in \Z$.
\begin{definition}
    Let $\{\psi_0,\dots,\psi_t\}$ be a system of affine linear forms. We say that the system has Cauchy-Schwarz complexity at most $s$ if for any $0\leq i\leq t$ one can partition $\{\psi_0,\dots,\psi_t\}\setminus \{\psi_i\}$ into $s+1$ classes (empty classes are allowed) such that $\psi_i$ does not lie in the affine linear span (over $\mathbb{Q}$) of the forms in any of these classes. The Cauchy-Schwarz complexity of the system is defined to be the least such $s$ or $\infty$ if no such $s$ exists.
\end{definition}
If $\psi\colon G^d\rightarrow G$ is an affine linear form, we denote by $\dot{\psi}\colon \mathbb{Q}^d\rightarrow \mathbb{Q}$ the map induced by its integer coefficients. The characteristic of $G$ is defined to be least order of all non-identity elements. Here is an equivalent formulation of Cauchy-Schwarz complexity in terms of change of variables.
\begin{proposition}\label{prop:equivalentCScomplexity}
Let $\{\psi_0,\dots,\psi_t\}$ be a system of affine linear forms $G^d\rightarrow G$. Suppose that the characteristic of $G$ is sufficiently large depending on the coefficients of $\psi_0,\dots,\psi_t$. Then the system has Cauchy-Schwarz complexity at most $s$ if and only if for every $0\leq i\leq t$ one can find a linear change of variables $\vec{x}=L_i(y_1,\dots,y_{s+1},z_1,\dots,z_d)$ on $\mathbb{Q}^d$ such that the form $\dot{\psi}_i(L_i(y_1,\dots,y_{s+1},z_1,\dots,z_d))$ has non-zero $y_1,\dots,y_{s+1}$ coefficients, but all other forms $\dot{\psi}_j(L_i(y_1,\dots,y_{s+1},z_1,\dots,z_d))$ with $j\neq i$ have at least one vanishing $y_1,\dots,y_{s+1}$ coefficient.
\end{proposition}
We need also the notion of non-classical polynomials, as this is important in the inverse Gowers theorem.
\begin{definition}
Let $V$ be a finite dimensional vector space over a finite field of characteristic $p$, $G$ an abelian group and let $P\colon V\rightarrow G$ be a map. We denote by $d_h$ the additive derivative, i.e. $d_hP(x):=P(x+h)-P(x)$. We say that $P$ is a non-classical polynomial of degree $\leq s$ if one has 
\begin{align*}
d_{h_1}\dots d_{h_{s+1}}P(x)=0,
\end{align*}
for all $h_1,\dots,h_{s+1}\in V$. We adopt the convention that the zero polynomial has degree $-\infty$. We denote the set of non-classical polynomial of degree $\leq s$ by $\text{Poly}_{\leq s}(V\rightarrow G)$.
\end{definition}
The abelian group $G$ will usually be $\T:=\R/\Z$. In this case we have the following proposition, see \cite{TaoZiegler:2012}.
\begin{proposition}
Let $P\in \text{Poly}_{\leq s}(V\rightarrow \T)$. Then there exists $\alpha\in \T$ such that $P$ takes values in the coset $\alpha+\frac{1}{p^{\lfloor\frac{s-1}{p-1}\rfloor+1}} \Z/ \Z$ of the $(p^{\lfloor\frac{s-1}{p-1}\rfloor+1})^{th}$ roots of unity, where $p$ is the characteristic of the ground field of $V$.
\end{proposition}
A consequence of the above proposition is that low degree polynomials, in comparison with the characteristic $p$, are polynomials in the classical sense (up to constants), i.e. $\tilde{P}\colon V\rightarrow \F_p$. We will now recall Gowers inverse theorem for vector spaces over finite fields.
\begin{theorem}\label{thm:GowersInverse}
(Gowers inverse) Let $V:=\F_p^n$ be a finite dimensional vector space over $\F_p$. Let $f\colon V\rightarrow\C$ be a function bounded in magnitude by 1 and also let $\varepsilon >0$. If $\Vert f\Vert_{U^{s+1}} \geq \varepsilon$, then there exists a non-classical polynomial $P\colon V\rightarrow \T$ of degree at most $s$ and a constant $\delta(\varepsilon,p,s)>0$ such that 
\begin{align*}
\vert\E_{x\in V} f(x)e(P(x))\vert \geq \delta(\varepsilon,p,s).
\end{align*} 
Here $e\colon \T\rightarrow\C:\theta\mapsto e^{2\pi i \theta}$.
\end{theorem}

\subsection{Results}
We will now continue with line games. Line games, as discussed in the introduction, fall inside a larger class of games which we will describe first. For this, let $\Gamma$ be a finite abelian group, let $m\geq 1$ an integer and we also have $t+1$ affine linear forms $\psi_0,\dots,\psi_t\colon\Gamma^m\rightarrow \Gamma$, i.e. 
\begin{align*}
\psi_i(g_1,\dots,g_m)=c_i+\sum_{j=1}^mc_{ij}g_j
\end{align*}
where $(g_1,\dots,g_m)\in \Gamma^m$, $c_i\in \Gamma$ and $c_{ij}\in \Z$.
\begin{definition}
A $t$-player linear forms game is given by the above data together with a game map $\rho\colon \Gamma\rightarrow \{0,1\}$ as follows. The referee samples a uniform random point $g$ from $\Gamma^m$ and sends $\psi_i(g)$ to player $i$ (players are numbered from 1 to $t$). The winning criterion is given by $\rho(\psi_0(g))$.
\end{definition}
Let $G$ be such a game. The classical bias is given by
\begin{align*}
\beta(G)=\max_{a_i\colon \Gamma\rightarrow \{\pm 1\}}\vert\E_{g\in \Gamma^m}(-1)^{\rho(\psi_0(g))}\prod_{i=1}^ta_i(\psi_i(g)) \vert.
\end{align*}
The quantum bias is
\begin{align*}
\beta^*(G)=\max_{N\geq 1, A_i\colon \Gamma\rightarrow \mathrm{Obs}^{\pm 1}(\C^N)}\Vert\E_{g\in \Gamma^m}(-1)^{\rho(\psi_0(g))}\prod_{i=1}^tA_i(\psi_i(g)) \Vert_{\text{op}}.
\end{align*}
Here $\mathrm{Obs}^{\pm 1}(\C^N)$ is the set of $\pm 1$-valued observables on $\C^N$.
\begin{remark}
It seems ad hoc that we consider such games, but these types of expressions are well studied in the context of counting linear patterns in finite abelian groups. We refer to \cite{tao2012higher}.
\end{remark}
The main technical theorem of this section is the following, from which we will deduce all other results.
\begin{theorem}
Let $G$ be a game as above. If the Cauchy-Schwarz complexity of $\{\psi_0,\dots,\psi_t\}$ is at most $s$, we then have the inequality
\begin{equation}
\beta^*(G)\leq \Vert (-1)^\rho\Vert_{U^{s+1}(\Gamma)}.
\end{equation}
\end{theorem}
To prove this theorem, we need the following lemma.
\begin{lemma}[Second Cauchy-Schwarz-Gowers inequality for operators]\label{lemma:2ndCSG}
Let $f\colon \Gamma\rightarrow \C$ be a function, $A_i\colon \Gamma^m\rightarrow \mathrm{End}(\C^N)$ for $i\in [m]$ such that $\Vert A_i(g) \Vert_{\mathrm{op}}\leq 1$ for any $g\in \Gamma^m$, $A_i$ is independent of the $i$-th coordinate of $g$ and $[A_i(g),A_j(h)]=0$, $[A_i(g)^*,A_j(h)]=0$ for all $i\neq j$ and $g,h\in\Gamma^m$. Then we have
\begin{align*}
\Vert \E_{(g_1,\dots,g_m)\in\Gamma^m} f(a_1g_1+\dots+a_mg_m)\prod_{i=1}^m A_i(g_1,\dots,g_m)\Vert_{\mathrm{op}}\leq \Vert f\Vert_{U^m(\Gamma)},
\end{align*}
where $a_i$ are non-zero integers such that the characteristic of $\Gamma$ exceeds all of them.
\end{lemma}
\begin{proof}
We will prove this by induction. For $m=1$ we have
\begin{align*}
\Vert\E_{g\in\Gamma}f(ag)A(g)\Vert_{\mathrm{op}} \leq \vert \E_{g\in\Gamma}f(ag)\vert=\vert \E_{g\in\Gamma}f(g)\vert=\Vert f\Vert_{U^1(\Gamma)}.
\end{align*}
Here we used that $A$ is independent of $g$. Assume we have proven the statement up to some integer $m\geq 1$. Then
\begin{align*}
\eta &:=  \Vert\E_{(g_1,\dots,g_{m+1})\in\Gamma^{m+1}}f(a_1g_1+\dots +a_{m+1}g_{m+1})\prod_{i=1}^{m+1}A_i(g_1,\dots,g_{m+1})\Vert_{\mathrm{op}}\\
&=\Vert \E_{(g_2,\dots,g_{m+1})\in\Gamma^{m}}A_1(g_2,\dots,g_{m+1})\E_{g_1\in\Gamma}f(a_1g_1+\dots +a_{m+1}g_{m+1})\prod_{i=2}^{m+1}A_i(g_1,\dots,g_{m+1})\Vert_{\mathrm{op}},
\end{align*}
we have done nothing, just rearranged and used the fact that $A_1$ is independent of $g_1$. Now write $F(g_2,\dots,g_{m+1}):=\E_{g_1\in\Gamma}f(a_1g_1+\dots +a_{m+1}g_{m+1})\prod_{i=2}^{m+1}A_i(g_1,\dots,g_{m+1})$ so that
\begin{align*}
\eta &= \Vert \E_{(g_2,\dots,g_{m+1})\in\Gamma^{m}}A_1(g_2,\dots,g_{m+1})F(g_2,\dots,g_{m+1})\Vert_{\mathrm{op}}\\
&\leq \Vert\E_{(g_2,\dots,g_{m+1})\in\Gamma^{m}}F(g_2,\dots,g_{m+1})F(g_2,\dots,g_{m+1})^*\Vert_{\mathrm{op}}^{1/2}.
\end{align*}
Here we used Proposition~\ref{prop:opcs} and we used the fact that $\Vert A_1(g)\Vert_{\mathrm{op}}\leq 1$ for any $g\in\Gamma^{m+1}$. Then
\begin{align*}
\eta &\leq \Vert \E_{g_1,g_1', g_2,\dots, g_{m+1}} f(a_1g_1+\dots +a_{m+1}g_{m+1})f(a_1g_1'+\dots +a_{m+1}g_{m+1})^*\\
&\qquad \times \prod_{i=2}^{m+1}A_i(g_1,\dots,g_{m+1})A_i(g_1',\dots,g_{m+1})^* \Vert_{\mathrm{op}}^{1/2}\\
&\leq (\E_{g_1,h_1}\Vert\E_{(g_2,\dots, g_{m+1})\in\Gamma^{m}} \Delta_{h_1} f(a_1g_1+\dots +a_{m+1}g_{m+1})\\
& \qquad \times \prod_{i=2}^{m+1}A_i(g_1+h_1,\dots,g_{m+1})A_i(g_1,\dots,g_{m+1})^*\Vert_{\mathrm{op}})^{1/2}\\
&\leq (\E_{g_1,h_1}(\E_{h_2,\dots,h_{m+1},z\in \Gamma}\Delta_{h_{m+1}}\dots\Delta_{h_1}f(a_1g_1+z))^{1/2^m})^{1/2} \\
& \leq  (\E_{g_1,h_1}(\E_{h_2,\dots,h_{m+1},z\in \Gamma}\Delta_{h_{m+1}}\dots\Delta_{h_1}f(a_1g_1+z)))^{1/2^{m+1}}\\
&= (\E_{h_1,h_2,\dots,h_{m+1},z\in \Gamma}\Delta_{h_{m+1}}\dots\Delta_{h_1}f(z))^{1/2^{m+1}}=\Vert f\Vert_{U^{m+1}(\Gamma)}.
\end{align*}
In the third line we used triangle inequality to get the expectation in $g_1,h_1$ outside the norm. In the fifth line we used the induction hypothesis to upper bound the expression in the previous line with the Gowers norm. We then use in the sixth line Jensens inequality. 
\end{proof}
\begin{proposition}[Generalized von Neumann inequality] \label{prop:GvN} Let $f\colon \Gamma\rightarrow\C$ be a function, $\{\psi_0,\dots,\psi_t\}$ a system of affine linear forms of Cauchy-Schwarz complexity $s$, $A_i\colon \Gamma^m\rightarrow \mathrm{End}(\C^N)$ for $i\in [t]$ such that $\Vert A_i(g) \Vert_{\mathrm{op}}\leq 1$ for any $g\in \Gamma^m$ and $[A_i(g),A_j(h)]=0$, $[A_i(g)^*,A_j(h)]=0$ for all $i\neq j$ and $g,h\in\Gamma^m$. Also assume the characteristic of $\Gamma$ is sufficiently large depending on the coefficients of the affine linear forms. Then we have the inequality
\begin{align*}
\Vert \E_{g\in\Gamma^m}f(\psi_0(g))\prod_{i=1}^tA_i(\psi_i(g))\Vert_{\mathrm{op}}\leq \Vert f\Vert_{U^{s+1}(\Gamma)}.
\end{align*}
\end{proposition}
\begin{proof}
The system of affine linear forms has Cauchy-Schwarz complexity $s$ so we can partition the forms $\{\psi_1,\dots,\psi_t\}$ into $s+1$ classes $\mathcal{A}_1,\dots,\mathcal{A}_{s+1}$ such that $\psi_0$ is not an affine linear combination of any forms in any class $\mathcal{A}_i$ for any $i$ (over $\mathbb{Q}$). So one can find a linear change of variables using Proposition~\ref{prop:equivalentCScomplexity}
\begin{align*}
(g_1,\dots,g_m)\mapsto (h_1,\dots,h_m)+y_1v_1+\dots +y_{s+1}v_{s+1}
\end{align*}
with the property that $\psi_0(y_jv_j)=a_jy_j$ where $a_j$ is a non-zero integer and $v_j\in \Z^m$, but if $\psi_i\in\mathcal{A}_j$, then $\psi_i(y_jv_j)=0$, this is where we need the large characteristic hypothesis. Now we define
\begin{align*}
\tilde{A}_k(g_1,\dots,g_m):=\prod_{j\in \mathcal{A}_k}A_j(\psi_j(g_1,\dots,g_m)).
\end{align*}
Note that $\tilde{A}_i$ is independent of its $i$-th coordinate and has operator norm smaller than 1. We then have
\begin{align*}
\Vert \E_{g\in\Gamma^m} & f(\psi_0(g_1,\dots,g_m)) \prod_{i=1}^{t}A_i(\psi_i(g_1,\dots,g_m))\Vert_{\mathrm{op}}\\
& = \Vert \E_{g\in\Gamma^m}f(\psi_0(g_1,\dots,g_m))\prod_{i=1}^{s+1}\tilde{A}_i(g_1,\dots,g_m)\Vert_{\mathrm{op}}\\
&= \Vert \E_{h\in \Gamma^m}\E_{y_1,\dots,y_{s+1}\in \Gamma}f(\psi_0(h)+a_1y_1+\dots +a_{s+1}y_{s+1})\prod_{i=1}^{s+1}\tilde{A}_i(h,y_1,\dots,y_{s+1})\Vert_{\mathrm{op}}\\
&\leq\E_{h\in \Gamma^m}\Vert f\Vert_{U^{s+1}(\Gamma)}=\Vert f\Vert_{U^{s+1}(\Gamma)}.
\end{align*}
In the third line we used the linear change of variables just described. Then we used the triangle inequality together with Lemma~\ref{lemma:2ndCSG} where we need the large characteristic hypothesis.
\end{proof}
\begin{remark}
If $f$ takes values in $\{\pm 1\}$ and $A_i$ are $\pm 1$-valued observables, then the inequality says that the quantum bias of such games is upperbounded by the Gowers norm of the game tensor $f$.
\end{remark}
Proposition~\ref{prop:GvN} is in full generality, i.e. for any abelian group the inequality holds. However, we will now restrict ourselves to the case where $\Gamma=\F_p^n$, where $p$ is prime and $n\geq 1$ as it will make many things easier. We can use Gowers inverse theorem for vector spaces over finite fields which we recalled in the preliminaries section. We will also assume that $p$ is sufficiently large, so that the set of non-classical polynomials coincide with the set of (classical) polynomials.
Let us start giving the proof of Theorem~\ref{thm:line_games}. A $t$-player line game is given by a map $\tau\colon \F_p^n\rightarrow \{0,1\}$ which stands for the predicate together with a system of linear forms $\psi_0,\psi_i\colon (\F_p^n)^2\rightarrow \F_p^n$ which are given by
\begin{align*}
\psi_0(x,y)=y \text{ and } \psi_i(x,y)=x+(i-1)y\text{ for }i=1,\dots,t.
\end{align*}
Note that the Cauchy-Schwarz complexity of this system is at most $t-1$.  For the bias of the game, it is more convenient to look at $f:=(-1)^\tau$. We also need the following lemma, provided kindly to us by Shravas Rao.
\begin{lemma}\label{lemma:Shravas-Jop}
Let $P\colon \F_p^n\rightarrow\F_p$ be a (classical) polynomial of degree $d-1$ and $p\geq d$. Then there exists $d$ polynomials $P_i\colon \F_p^n\rightarrow\F_p$, $i=0,\dots,d-1$, such that
\begin{align*}
P(y)=\sum_{i=0}^{d-1}P_i(x+iy).
\end{align*}
\end{lemma}
\begin{proof}
The polynomial $P$ can be represented as
\begin{align*}
P(x_1,\dots,x_n)=\sum_{i=0}^{d-1} T_i(x,\dots, x),\quad x=(x_1,\dots,x_n),
\end{align*}
where each $T_i\colon (\F_p^n)^{i}\rightarrow \F_p$ is an $i$-linear form. We will show that for each linear form $T_i$ we can find $\alpha_0, \dots, \alpha_{d-1}$ such that
\begin{align}\label{eqn:conditionEquality}
T_i(y,\dots,y)=\sum_{j=0}^{d-1}\
\alpha_jT_i(x+jy,\dots,x+jy)
\end{align}
and this will be enough to construct $P_0,\dots,P_{d-1}$. By linearity, we can rewrite the right hand side as follows,
\begin{align*}
\sum_{j=0}^{d-1} & \sum_{s\in\{0,1\}^i}\alpha_jT_i((1-s_1)x+s_1jy,\dots, (1-s_i)x+s_ijy)\\
& = \sum_{j=0}^{d-1}\sum_{s\in\{0,1\}^i}\alpha_jj^{|s|}T_i((1-s_1)x+s_1y,\dots, (1-s_i)x+s_iy),
\end{align*}
where $|s|$ denotes the Hamming weight of $s$. Then \ref{eqn:conditionEquality} holds, if for $0\leq k< i$
\begin{align*}
\sum_{j=0}^{d-1}\alpha_jj^k=0\text{ and }\sum_{j=0}^{d-1}\alpha_jj^i=1.
\end{align*}
As $d\leq p$ the $d\times d$ Vandermonde matrix associated with the sequence $1,j,\dots, j^i$ is invertible, hence there exist unique $\alpha_0,\dots,\alpha_{d-1}$ satisfying the above equations which concludes the proof.
\end{proof}

The following lemma will help us later in converting complex strategies into $\pm 1$-strategies.
\begin{lemma}\label{lemma:neatDescriptionOfComplexUnit} For any $z\in \C$ 
\begin{align*}
z=\frac{\pi}{2}\E_{\vert w \vert =1}[\sign (z\overline{w}) \; \vert z\vert \; w],
\end{align*}
where $w\in \{z\in \C\colon \vert z\vert =1\}$ is taken uniformly at random.
\end{lemma}
\begin{proof}
Write $z=re^{i\psi}$. Then
\begin{align*}
\E_{\vert w \vert =1}[\sign (z\overline{w})\vert z\vert w]&=\frac{r}{2\pi}\int_{0}^{2\pi}\sign (e^{-i(\phi-\psi)})e^{i\phi} d \phi\\
&=\frac{r}{2\pi}\int_0^{2\pi}\sign (e^{-i\chi})e^{i\chi+i\psi} d \chi\\
&=\frac{z}{2\pi}\int_0^{2\pi}\sign (e^{-i\chi})e^{i\chi} d \chi\\
&=\frac{2z}{\pi}.
\end{align*}
\end{proof} 

\begin{proof}[Proof of Theorem~\ref{thm:line_games}] By Proposition~\ref{prop:GvN} and the hypothesis that the game has entangled value $\varepsilon>0$ implies that $\Vert f \Vert_{U^t}>\varepsilon$. Then by Gowers inverse Theorem~\ref{thm:GowersInverse} and assumption that $p>t$ there exists a constant $\delta = \delta(\varepsilon,p,t)>0$ and a (classical) polynomial of degree at most $t-1$ such that
\begin{align*}
\vert \E_{x\in\F_p^n}f(x)e(P(x))\vert > \delta.
\end{align*}
We now want to convert this presence of structure into a classical strategy. First by Lemma~\ref{lemma:Shravas-Jop} we can find $t$ polynomials $P_i$ for $i=1,\dots,t$ such that
\begin{align*}
P(y)=\sum_{i=1}^{t}P_i(x+(i-1)y).
\end{align*}
This implies
\begin{align*}
\vert \E_{x,y\in \F_p^n}f(\psi_0(x,y))e(P(\psi_0(x,y)))\vert = \vert \E_{x,y\in \F_p^n}f(\psi_0(x,y))\prod_{i=1}^{t} e(P_i(\psi_i(x,y)))\vert>\delta.
\end{align*}
The polynomials are not classical strategies yet, we can turn it into $\pm 1$-strategy using Lemma~\ref{lemma:neatDescriptionOfComplexUnit} at a loss of a factor $2^t/\pi^t$.
\end{proof}

\subsection{Parallel repetition}
Let $f\colon \Gamma\rightarrow\{1,-1\}$ be a function, representing the predicate. We want to consider $k$-fold XOR parallel repetition. The predicate for this is $f^k\colon\Gamma^k\rightarrow \{1,-1\}$ defined by
\begin{align*}
f^k(g_1,\dots,g_k):=\prod_{i=1}^kf(g_i).
\end{align*}
\begin{lemma}\label{lemma:GowersNormProductFunction}
We have that
\begin{align*}
\Vert f^k\Vert_{U^{s+1}(\Gamma^k)}=\Vert f\Vert_{U^{s+1}(\Gamma)}^k.
\end{align*}
\end{lemma}
\begin{proof}
This follows immediately from the definition of Gowers norm.
\end{proof}
Let $\{\psi_0,\dots,\psi_t\}$ be linear forms $\Gamma^m\rightarrow \Gamma$ which together with $f$ define the game $G$. The linear forms corresponding with $k$-fold XOR parallel repetition are denoted by $\{\psi_0^k,\dots,\psi_t^k\}$ which are maps $(\Gamma^m)^k\rightarrow\Gamma^k$ and are given by
\begin{align*}
\psi_i^k(g^1,\dots,g^k):=(\psi_i(g^1),\dots,\psi_i(g^k)),\text{ where }g^i\in \Gamma^m.
\end{align*}
Note that if $\{\psi_0,\dots,\psi_t\}$ has Cauchy-Schwarz complexity at most $s$, then the Cauchy-Schwarz complexity of $\{\psi_0^k,\dots,\psi_t^k\}$ is also at most $s$. Denote by $G^{\oplus k}$ the $k$-fold XOR parallel repetition, then we have as an immediate consequence of Proposition~\ref{prop:GvN} together with Lemma~\ref{lemma:GowersNormProductFunction} the following upper bound 
\begin{align*}
\beta^*(G^{\oplus k})\leq \Vert f\Vert^k_{U^{s+1}(\Gamma)}. 
\end{align*}
If $G$ is an XOR game, denote by $G^k$ the $k$-fold parallel repetition. If $S$ is a strategy (classical or quantum) for a game $G$, denote by $\omega(G,S)$ the winning probability using strategy $S$. Also denote by $\varepsilon(G,S):=2\omega(G,S)-1$ the bias of this strategy. To prove Lemma~\ref{lemma:parallelRepetition}, we use the following lemma, which is a straightforward generalization of the 2-player version in \cite{cleve2008perfect} (lemma 8 in that paper) to any number of players. 
\begin{lemma}\label{lemma:lemmaCleveSlofstra}
Let $G$ be an XOR game assume. Let $S$ be any strategy for $G^k$. For each $M\subset [k]$, we denote by $S_M$ the following strategy for the XOR parallel repetition $\oplus_{i\in M}G$ : (1) Run strategy $S$, yielding answers $a_i^1,\dots,a_i^k$ for player $i=1,\dots, t$. (2) Player $i$ outputs $\sum_{j\in M}a_i^j\mod 2$. We then have
\begin{align*}
\omega(G^k,S)=\frac{1}{2^k}\sum_{M\subset [k]}\varepsilon(\oplus_{i\in M}G,S_M).
\end{align*}
\begin{proof}[Proof of Lemma~\ref{lemma:parallelRepetition}.] Let $S$ be the quantum strategy that achieves the maximum winning probability of the game $G^k$. We then use Lemma~\ref{lemma:lemmaCleveSlofstra},
\begin{align*}
\omega^*(G^k)&=\omega(G^k,S)=\frac{1}{2^k}\sum_{M\subset [k]}\varepsilon(\oplus_{i\in M}G,S_M)\\
&\leq \frac{1}{2^k}\sum_{M\subset [k]}\beta^*(G^{\oplus \vert M\vert})=\frac{1}{2^k}\sum_{l=0}^k\beta^*(G^{\oplus l}) {k \choose l}\\
&\leq \frac{1}{2^k}\sum_{l=0}^k {k \choose l} \Vert f\Vert_{U^{s+1}(\Gamma)} =\left(\frac{1+\Vert f\Vert_{U^{s+1}(\Gamma)}}{2}\right)^k.
\end{align*}
\end{proof}
\end{lemma}

\section{Near-perfect strategies for 2-player unique games}
\label{sec:perfectunique}

In this section we prove Theorem~\ref{thm:uniquegames}.
Consider a unique game where $\pi_{xy}$ is the matching between the players' answers on inputs $x,y$, so that when the first player answers $i$ they win if the second player answers $j=\pi_{xy}(i)$.
Let us start by writing down an expression for the entangled winning probability when the players use a shared state $\ket{\psi}$ and projectors $\Pi^{(x)}_i, \Pi^{(y)}_j$ for inputs $x,y$ and outputs $i,j$. For finite-dimensional systems we can always assume that a strategy is of such a form. The winning probability is given by
\begin{align*}
    \E_{x,y} \sum_{i=1}^k \Pr(\text{answer }\; i,\pi_{xy}(i) \mid \text{input }x,y) = \E_{x,y} \sum_{i=1}^k \bra{\psi} \Pi^{(x)}_i \otimes \Pi^{(y)}_{\pi_{xy}(i)} \ket{\psi}.
\end{align*}
Now define vectors $\ket{u^{(x)}_i} = (\Pi^{(x)}_i \otimes \mathrm{Id}) \ket{\psi}$ and $\ket{v^{(y)}_j} = (\mathrm{Id} \otimes \Pi^{(y)}_j )\ket{\psi}$, then we can write the winning probability as
\begin{align} \label{eq:winprob1}
    \E_{x,y} \sum_{i=1}^k \langle u^{(x)}_i \vert v^{(y)}_{\pi_{xy}(i)} \rangle \geq 1 - \epsilon
\end{align}
where we use the assumption that there is a strategy with entangled value at least $1-\epsilon$.
The vectors have the following properties:
\begin{align*}
    \forall x,y , \forall i \neq j && \langle u^{(x)}_i \vert u^{(x)}_j \rangle = \langle v^{(y)}_i \vert v^{(y)}_j \rangle = 0 \tag{orthogonal projectors}\\
    \forall x,y && \sum_{i=1}^k \norm{u^{(x)}_i}^2 = \sum_{i=1}^k \norm{v^{(y)}_i}^2 = 1 \tag{projectors sum to identity}\\
    \forall x,y , \forall i, j && \langle u^{(x)}_i \vert v^{(y)}_j \rangle \geq 0 \tag{projectors are Hermitian}
\end{align*}
By using $\norm{u-v}^2 = \norm{u}^2 + \norm{v}^2 - 2\braket{u}{v}$ (for real-valued inner products) we can write~\eqref{eq:winprob1} as
\begin{align} \label{eq:winprob2}
    \frac{1}{2} \E_{x,y} \sum_{i=1}^k \norm{ u^{(x)}_i - v^{(y)}_{\pi_{xy}(i)} }^2 \leq \epsilon .
\end{align}
It is possible to maximize expression \eqref{eq:winprob1} (or equivalently minimize \eqref{eq:winprob2}) over vectors with the given properties. This optimization problem is an SDP and can be solved in polynomial time but will generally not yield a quantum strategy as not all such vectors can be attained by quantum strategies. Our goal will be to extract from the vectors a classical strategy, something known as rounding, such that its winning probability is high.

As stated in the introduction, one can get some intuition by considering the $\epsilon=0$ case. There~\eqref{eq:winprob2} yields $\ket{u^{(x)}_i} = \ket{v^{(y)}_{\pi_{xy}(i)}}$ for each $x,y$ and $i$. Using shared randomness the players sample a random vector $\ket{g}$ and compute the overlaps $\xi^{(x)}_i = \langle g \ket{u^{(x)}_i}$ and $\xi^{(y)}_i = \langle g \ket{v^{(y)}_i}$ respectively. The players will have the same values $\xi^{(x)}_i = \xi^{(y)}_{\pi_{xy}(i)}$ so both players can output the answer $i$ for which their overlap has the largest value.

For $\epsilon>0$ the sets of vectors will not be exactly equal and therefore the values $\xi^{(x)}_i,\xi^{(y)}_{\pi_{xy}(i)}$ will be close but not exactly equal. The discrepancy in these values will be bigger for vectors $\ket{u^{(x)}_i}$ with a small norm. In Section 2 of~\cite{Charikar:2006} a rounding algorithm is provided that solves these issues.
Note that we write $u^{(x)}_i, v^{(y)}_j$ for the vectors belonging to questions $x,y$ and answers $i,j$ whereas in~\cite{Charikar:2006} these vectors are instead denoted by $u_i,v_j$ where $u,v$ are the questions and $i,j$ the answers. The only difference between their SDP and the above one is that they have an additional constraint $0 \leq \langle u^{(x)}_i \vert v^{(y)}_{\pi_{xy}(i)} \rangle \leq |u^{(x)}_i|^2$ (constraint (5) in their paper). This constraint does not necessarily hold in the quantum setting so we will drop it and adapt their proofs to work without this constraint. 

The following is \textbf{Rounding Algorithm 2} from Section 4 of~\cite{Charikar:2006}, adapted to our notation.\\
\textbf{Rounding algorithm}\\
\textbf{Input:} A solution of the SDP with objective value $1-\epsilon$.\\
\textbf{Output:} A classical strategy: $a(x)$ and $b(y)$\\
Define $[x]_r$ as the function that rounds $x$ up or down depending on whether the fractional part of $x$ is greater or less than $r$. If $r$ is uniform random on $[0,1]$ then the expected value of $[x]_r$ is $x$.
\begin{enumerate}
    \item \label{step:definetilde} Define $\ket{\tilde{u}^{(x)}_i} = \ket{u^{(x)}_i} / \norm{u^{(x)}_i}$ if $\norm{u^{(x)}_i} \neq 0$, otherwise $\ket{\tilde{u}^{(x)}_i} = 0$.
    \item \label{step:pickr} Pick $r \in [0,1]$ uniformly at random.
    \item Pick random independent Gaussian vectors $\ket{g_1},...,\ket{g_{2k}}$ with independent components distributed as $\mathcal{N}(0,1)$.
    \item For each question $x$:
        \begin{enumerate}
            \item Set $s^{(x)}_i = \left[2k\cdot \norm{u^{(x)}_i}^2 \right]_r$. \label{step:picks}
            \item For each $i$ project $s^{(x)}_i$ vectors $\ket{g_1},...,\ket{g_{s^{(x)}_i}}$ to $\ket{\tilde{u}^{(x)}_i}$:
                \begin{align*}
                    \xi^{(x)}_{i,s} = \langle g_s \vert \tilde{u}^{(x)}_i \rangle , \quad s = 1,2,..., s^{(x)}_i
                \end{align*}
            \item Select the $\xi^{(x)}_{i,s}$ with the largest absolute value. Assign $a(x) = i$.
        \end{enumerate}
    \item Repeat the previous step for each question $y$ but with the vectors $\ket{v^{(y)}_j}$ to obtain $b(y)$.
\end{enumerate}

The intuition behind the algorithm is as follows. Similar to the $\epsilon=0$ case, the values $\xi^{(x)}_{i,s}$ and $\xi^{(y)}_{\pi_{xy}(i),s}$ will be close.
Vectors $\ket{u^{(x)}_i}$ and $\ket{u^{(x)}_j}$ for different answers $i\neq j$ are orthogonal and their corresponding values $\xi^{(x)}_{i,s}$ and $\xi^{(x)}_{j,s}$ are therefore independent.
For vectors with small norm, the values $\xi^{(x)}_{i,s}$ and the matching $\xi^{(x)}_{\pi_{xy}(i),s}$ will be less correlated. Therefore we sample more Gaussian vectors for answers corresponding to a high norm (step~\ref{step:picks}).

To prove Theorem~\ref{thm:uniquegames} we have to show that the result of the above rounding algorithm is a strategy with winning probability $1-\mathcal{O}(\sqrt{\epsilon \log k})$. This is exactly the result of Theorem 4.5 of~\cite{Charikar:2006} with the exception of the additional constraint mentioned before. This modification requires a different proof of Lemma 4.2 and 4.3 in~\cite{Charikar:2006} but leaves the remaining part of their proof unchanged. We therefore only prove these Lemma's and refer the reader to Section 4 of~\cite{Charikar:2006} for the remainder of the proof.

We adopt their definitions
\begin{align*}
    \epsilon_{xy}   &= \frac{1}{2}\sum_{i=1}^k \norm{ u^{(x)}_i - v^{(y)}_{\pi_{xy}(i)} }^2 ,\\
    \epsilon_{xy}^i &= \frac{1}{2}\norm{ \tilde{u}^{(x)}_i - \tilde{v}^{(y)}_{\pi_{xy}(i)} }^2 ,
\end{align*}
where $\tilde{u}^{(x)}_i$ and $\tilde{v}^{(y)}_i$ were defined in step~\ref{step:definetilde} of the rounding algorithm. Note that $\E_{x,y} \epsilon_{xy} \leq \epsilon$. 

\begin{lemma}[Originally Lemma 4.2] \label{lem:CharikarReproven}
    The probability that the rounding algorithm gives a correct assignment to the questions $x,y$ is $1-\mathcal{O}(\sqrt{\epsilon_{xy} \log k})$.
\end{lemma}
\begin{proof}
    If $\epsilon_{xy} \geq 1/128$ then the statement follows trivially since this can be hidden in the big-O. Therefore assume $\epsilon_{xy} \leq 1/128$. Define
    \begin{align*}
        M   &= \left\{ (i,s) \;:\; i\in [k] \;,\; s \leq \min( s^{(x)}_i , s^{(y)}_{\pi_{xy}(i)} ) \right\} ,\\
        M_c &= \left\{ (i,s) \;:\; i\in [k] \;,\; \min( s^{(x)}_i , s^{(y)}_{\pi_{xy}(i)} ) < s \leq \max( s^{(x)}_i , s^{(y)}_{\pi_{xy}(i)} ) \right\} .
    \end{align*}
    The set $M$ contains the pairs $(i,s)$ for which both $\xi^{(x)}_{i,s}$ and $\xi^{(y)}_{\pi_{xy}(i),s}$ are defined and the set $M_c$ contains the pairs for which only one of these is defined.

    We need the following two lemmas to continue.
    \begin{lemma} \label{lemma:sizeofM}
        When $\epsilon_{xy} \leq 1/128$ then $\E_r[ |M_c| ] \leq 4k \sqrt{2\epsilon_{xy}}$ and $|M| \geq k/2$.
    \end{lemma}
    This was originally Lemma 4.3 and it stated $E_r[ |M_c| ] \leq 4k \epsilon_{xy}$.
    \begin{proof}
        The expected value of $|s^{(x)}_i - s^{(y)}_{\pi_{xy}(i)}|$ is given by
        \begin{align*}
            \E_r \left\vert \left[2k\cdot \norm{u^{(x)}_i}^2\right]_r - \left[2k\cdot \norm{v^{(y)}_{\pi_{xy}(i)}}^2\right]_r \right\vert = 2k \left\vert \norm{u^{(x)}_i}^2 - \norm{v^{(y)}_{\pi_{xy}(i)}}^2 \right\vert .
        \end{align*}
        By the triangle inequality
        \begin{align*}
            \left\vert \norm{u^{(x)}_i}^2 - \norm{v^{(y)}_{\pi_{xy}(i)}}^2 \right\vert
            &=
            \left\vert \norm{u^{(x)}_i} - \norm{v^{(y)}_{\pi_{xy}(i)}} \right\vert
            \left( \norm{u^{(x)}_i} + \norm{v^{(y)}_{\pi_{xy}(i)}} \right) \\
            &\leq
            \norm{u^{(x)}_i - v^{(y)}_{\pi_{xy}(i)}}
            \left( \norm{u^{(x)}_i} + \norm{v^{(y)}_{\pi_{xy}(i)}} \right) ,
        \end{align*}
        and by using Cauchy-Schwarz twice we have
        \begin{align*}
            \sum_{i=1}^k \norm{u^{(x)}_i - v^{(y)}_{\pi_{xy}(i)}} \left( \norm{u^{(x)}_i} + \norm{v^{(y)}_{\pi_{xy}(i)}} \right)
            &\leq \sqrt{\sum_{i=1}^k \norm{u^{(x)}_i - v^{(y)}_{\pi_{xy}(i)}}^2 } \sqrt{\sum_{i=1}^k \left( \norm{u^{(x)}_i} + \norm{v^{(y)}_{\pi_{xy}(i)}} \right)^2 } \\
            &= \sqrt{2\epsilon_{xy}} \sqrt{\sum_{i=1}^k \left( \norm{u^{(x)}_i}^2 + \norm{v^{(y)}_{\pi_{xy}(i)}}^2 + 2 \norm{u^{(x)}_i}\norm{v^{(y)}_{\pi_{xy}(i)}} \right) } \\
            &\leq \sqrt{2\epsilon_{xy}} \sqrt{1 + 1 + 2} = 2\sqrt{2\epsilon_{xy}} .
        \end{align*}
        The proof follows from $|M_c| = \sum_{i=1}^k |s^{(x)}_i - s^{(y)}_{\pi_{xy}(i)}|$.
        For the second part of the lemma, observe that
        \begin{align*}
            \min(s^{(x)}_i, s^{(y)}_{\pi_{xy}(i)})
            &\geq 2k \; \min(\norm{u^{(x)}_i}^2 , \norm{v^{(y)}_{\pi_{xy}(i)}}^2 ) - 1 \\
            &\geq 2k\left( \norm{u^{(x)}_i}^2 - \left\vert \norm{u^{(x)}_i}^2 - \norm{v^{(y)}_{\pi_{xy}(i)}}^2 \right\vert \right) - 1.
        \end{align*}
        Therefore we have
        \begin{align*}
            |M| = \sum_{i=1}^k \min(s^{(x)}_i, s^{(y)}_{\pi_{xy}(i)})
            &\geq \sum_{i=1}^k \left( 2k \;\norm{u^{(x)}_i}^2 - 2k\;\left\vert \norm{u^{(x)}_i}^2 - \norm{v^{(y)}_{\pi_{xy}(i)}}^2 \right\vert  - 1 \right) \\
            &\geq 2k - 4k\sqrt{2\epsilon_{xy}} - k \geq k/2,
        \end{align*}
        where we used $\epsilon_{xy} \leq 1/128$.
    \end{proof}
    \begin{lemma}\label{lemma:averageepsilon}
        The following inequality holds
        \begin{align*}
            \E_r \left[ \frac{1}{|M|} \sum_{(i,s)\in M} \epsilon_{xy}^i \right] \leq 4 \epsilon_{xy}
        \end{align*}
    \end{lemma}
    \begin{proof}
        This is Lemma 4.4 in~\cite{Charikar:2006}.
    \end{proof}
    We now continue the proof of Lemma~\ref{lem:CharikarReproven}.
    First consider a fixed value of $r$ (picked in step~\ref{step:pickr} of the rounding algorithm.
    Consider the sequences $\xi^{(x)}_{i,s}$ and $\xi^{(y)}_{\pi_{xy}(i),s}$ where the indices $(i,s)$ run over all $(i,s)\in M$. We apply Theorem 4.1 of~\cite{Charikar:2006} to these sequences and get that the probability that the largest absolute value in the first sequence has the same index as the largest absolute value in the second sequence is
    \begin{align*}
        1 - \mathcal{O}\left(\sqrt{ \log|M| \cdot \frac{1}{|M|} \sum_{(i,s)\in M} \epsilon^{i}_{xy} }\right) .
    \end{align*}
    By Jensen's inequality we have
    \begin{align*}
        \E_r \left[ 1 - \mathcal{O}\left(\sqrt{ \log|M| \cdot \frac{1}{|M|} \sum_{(i,s)\in M} \epsilon^{i}_{xy} }\right) \right]
        &\geq 1-\mathcal{O}\left( \sqrt{ \E_r\left[ \log|M| \cdot \frac{1}{|M|} \sum_{(i,s)\in M} \epsilon^{i}_{xy} \right] } \right) \\
        &\geq 1-\mathcal{O}\left( \sqrt{ \epsilon_{xy} \log k } \right)
    \end{align*}
    where the second inequality follows from $|M| \leq 3k$ and Lemma~\ref{lemma:averageepsilon}.
    In the rounding algorithm, the largest $\xi^{(x)}_{i,s}$ is picked not only among the $(i,s)\in M$ but also $(i,s)\in M_c$. However, the probability that the index for the largest value is in $M_c$ is at most
    \begin{align*}
        \E_r \left[ \frac{|M_c|}{|M|} \right] \leq \frac{4k\sqrt{2\epsilon_{xy}}}{k/2} = 8\sqrt{2\epsilon_{xy}} ,
    \end{align*}
    by Lemma~\ref{lemma:sizeofM}. Therefore by the union bound, the probability that the answers match is at least
    \begin{align*}
        1-\mathcal{O}(\sqrt{\epsilon_{xy}\log k}) - 8 \sqrt{2\epsilon_{xy}} = 1-\mathcal{O}(\sqrt{\epsilon_{xy}\log k}) ,
    \end{align*}
    which finishes the proof.
\end{proof}
With these modified lemmas, Theorem 4.5 of~\cite{Charikar:2006} shows that the rounding algorithm gives a classical strategy that wins the game with probability $1-\mathcal{O}(\sqrt{\epsilon\log k})$.

\section*{Acknowledgements}

We thank Peter H\o yer, Serge Massar, and Henry Yuen for useful discussions. We thank Shravas Rao for providing a proof of one of the lemmas.

\bibliographystyle{alpha}
\bibliography{main}

\end{document}